%% file: toplas.tex
\documentclass[acmsmall, screen, nonacm]{acmart}

\input{preamble}

\usepackage[scaled=0.85]{beramono}   
\usepackage{listings}
\newcommand{\saml}[1]{\todo[inline]{SL: #1}}
\usepackage{color}

\usepackage{bbding}      
\usepackage{orcidlink}

\setcopyright{acmlicensed}
\copyrightyear{2018}
\acmYear{2018}
\acmDOI{XXXXXXX.XXXXXXX}

\acmJournal{JACM}
\acmVolume{37}
\acmNumber{4}
\acmArticle{111}
\acmMonth{8}

\begin{document}

\title{Scoped Effects as Parameterized Algebraic Theories}

\author{Cristina Matache}
\author{Sam Lindley}
\affiliation{%
  \institution{University of Edinburgh}
  \country{UK}
}

\author{Sean Moss}
\affiliation{%
  \institution{University of Birmingham}
  \country{UK}
}

\author{Sam Staton}
\affiliation{%
  \institution{University of Oxford}
  \country{UK}
}

\author{Nicolas Wu}
\author{Zhixuan Yang}
\affiliation{%
 \institution{Imperial College London}
 \country{UK}
}

\renewcommand{\shortauthors}{Matache et al.}

\begin{abstract}
  Notions of computation can be modelled by monads. \emph{Algebraic
  effects} offer a characterization of monads in terms of algebraic
  operations and equational axioms, where operations are basic
  programming features, such as reading or updating the state, and
  axioms specify observably equivalent expressions. However, many
  useful programming features depend on additional mechanisms such as
  delimited scopes or dynamically allocated resources. Such mechanisms
  can be supported via extensions to algebraic effects including
  \emph{scoped effects} and \emph{parameterized algebraic
  theories}. We present a fresh perspective on scoped effects by
  translation into a variation of parameterized algebraic
  theories. The translation enables a new approach to equational
  reasoning for scoped effects and gives rise to an alternative
  characterization of monads in terms of generators and equations
  involving both scoped and algebraic operations. We demonstrate the
  power of our approach by way of equational
  characterizations of several known models of scoped effects.
\end{abstract}




\maketitle

\section{Introduction}\label{sec:introduction}

The central idea of \emph{algebraic effects}~\cite{DBLP:conf/fossacs/PlotkinP01} is that impure computation can be built and reasoned about equationally, using an algebraic theory.
\emph{Effect handlers}~\cite{PlotPret13Hand} are a way of implementing algebraic effects and provide a method for modularly programming with different effects.
More formally, an effect handler gives a model for an algebraic theory.
In this paper we develop equational reasoning for a notion arising from an extension of handlers, called \emph{scoped effects}, using the framework of \emph{parameterized algebraic theories}.

The central idea of \emph{scoped effects} (Sec.~\ref{sec:scop-effects-handl}) is that certain parts of an impure computation should be dealt with one way, and other parts another way, inspired by scopes in exception handling.
Compared to algebraic effects, the crucial difference is that the scope on which a scoped effect acts is delimited.
This difference leads to a complex relationship with monadic sequencing ($\bind$).
The theory and practice of scoped effects~\cite{WSH2014,DBLP:conf/lics/PirogSWJ18,YPWvS2021,Bosman2023calculus,Berg2023framework,YW2023} has primarily been studied by extending effect handlers to deal with not just algebraic operations, but also more complex scoped operations.
They form the basis of the \texttt{fused-effects} and \texttt{polysemy} libraries for Haskell.
Aside from exception handling, other applications include back-tracking in parsing~\cite{WSH2014} and timing analysis in telemetry~\cite{github}.

\emph{Parameterized algebraic theories} (Sec.~\ref{sec:pat}) extend plain algebraic theories with variable binding operations for an abstract type of parameters. They have been used to study various resources including logic variables in logic programming~\cite{DBLP:conf/fossacs/Staton13}, channels in the $\pi$-calculus~\cite{DBLP:conf/lics/Staton13}, code pointers~\cite{jumps}, qubits in quantum programming~\cite{DBLP:conf/popl/Staton15}, and urns in probabilistic programming~\cite{beta-bernoulli}.


\paragraph{Contributions.}
%
We propose an equational perspective for scoped effects where \emph{scopes are resources},
%
by analogy with other resources like file handles.
We develop this perspective using the framework of \emph{parameterized algebraic theories}, which provides an algebraic account of effects with resources and instances.
We realize scoped effects by encoding the scopes as resources with open/close operations, analogous to opening/closing files.
This paper provides:
\begin{itemize}
\item the first syntactic sound and complete equational reasoning system for scoped effects, based on the equational reasoning for parameterized algebraic theories (Prop.~\ref{thm:sound}, Prop.~\ref{prop:free-models});
\item a canonical notion of semantic model for scoped effects supporting three key examples from the literature: nondeterminism with semi-determinism (Thm.~\ref{thm:nondet-free-model}), catching exceptions (Thm.~\ref{thm:catch-free-model}), local state (Thm.~\ref{thm:state-free-model}), and nondeterminism with cut (Thm.~\ref{thm:cut-free-model}); and
\item a reconstruction of the previous categorical analysis of scoped effects via parameterized algebraic theories: the constructors $(\earlier,\later$) are shown to be not ad hoc, but rather the crucial mechanism for arities/coarities in parameterized algebraic theories (Thm.~\ref{thm:noeqn}).
\end{itemize}


\paragraph{Example: nondeterminism with semi-determinism.}
%
 \begin{wrapfigure}[10]{r}{0.26\textwidth}
       \begin{tikzpicture}
    \node(once){} [level distance=0.55cm, sibling distance=1.3cm]
    child{node(or1){$\orop$} 
      child{node(fail){$\fail$}}
      child{node(or2){$\orop$}
        child{node{$\orop$} [sibling distance=0.8cm]
          child{node{$1$}}
          child{node{$2$}}
        }
        child{node{$\orop$} [sibling distance=0.8cm]
          child{node{$3$}}
          child{node{$4$}}
        }
      }
    }
    ++(0,-2mm) node (ghost) {};
    \node[draw,red,fit=(ghost.mid)(fail)(or2.north)(or2.east)]{};
    \node[fill=white] at (once) {$\once$};
  \end{tikzpicture}
\caption{Illustrating~(\ref{eq:scoped-once}) }
\label{fig:scoped-once}
\hide{\begin{minipage}{0.45\textwidth}
\begin{center}
  \begin{tikzpicture}
    \node(once){$\once$} [level distance=0.6cm, sibling distance=1.5cm]
    child{node(or1){$\orop_a$}[level distance = 0.65cm]
      child{node(fail){$\fail_a$}}
      child{node(or2){$\orop_a$}
        child{node{$\close_a$}
        child{node{$\orop$} [level distance=0.5cm, sibling distance=1cm]
          child{node{$1$}}
          child{node{$2$}}
        }}
        child{node{$\close_a$}
        child{node{$\orop$} [level distance=0.5cm, sibling distance=1cm]
          child{node{$3$}}
          child{node{$4$}}
        }}
      }
    };
  \end{tikzpicture}
\caption{Term~(\ref{eq:param-once}) as a tree}
\label{fig:param-once}
\end{center}
\end{minipage}}
 \end{wrapfigure}
We now briefly illustrate the intuition underlying the connection between scoped effects and parameterized algebraic theories through an example. (See Examples~\ref{ex:sig:nondet}~and~\ref{ex:nondet-eq} for further details.)
Let us begin with two algebraic operations: $\orop(x,y)$, which nondeterministically chooses between continuing\footnote{
This continuation-passing style is natural for algebraic effects, but when programming one often uses equivalent direct-style \emph{generic effects}~\cite{pp-alg-op-gen-eff} such as $\underline{\orop}:\mathsf{unit}\rightarrow\mathsf{bool}$, where $\orop(x, y)$ can be recovered by pattern matching on the result of $\underline{\orop}$.
} as computation $x$ or as computation $y$, and $\fail$, which fails immediately.
%
%
%
We add semi-determinism in the form of a \emph{scoped} operation $\once(x)$, which chooses the first branch of the computation $x$ that does not fail.
Importantly, the scope that $\once$ acts on is delimited.
The left program below returns $1$; the right one returns $1$ or $2$, as the second $\orop$ is outside the scope of $\once$.
\[
  \once(\orop(\orop(1,2),\orop(3,4)))
  \qquad
  \once(\orop(1,3))\bind \lambda x.\,\orop(x,x+1)
\]

Now consider a slightly more involved example, which also returns $1$
or $2$:
\begin{equation}
  \label{eq:scoped-once}
  \once(\orop(\fail,\,\orop(1,3)))\bind \lambda x.\,\orop(x,x+1)
\end{equation}
depicted as a tree in Fig.~\ref{fig:scoped-once} where the red box delimits the scope of $\once$.
%
%
%
We give an encoding of term~(\ref{eq:scoped-once}) in a parameterized algebraic theory as follows:
\begin{equation}
  \label{eq:param-once}
  \once(a.\orop(\fail, \orop(\close(a,\orop(1, 2)), \close(a,\orop(3, 4)))))
\end{equation}
where $a$ is the name of the scope opened by $\once$ and closed by the special $\close$ operation.
%
By equational reasoning for scoped effects~(\S\ref{sec:param:theories:scoped}) and the equations for nondeterminism (Fig.~\ref{fig:nondet-once}), we can prove that the term~(\ref{eq:param-once}) is equivalent to~$\orop(1, 2)$.

\paragraph{Changes from the conference version.}
This paper is an extended version of a ``fresh perspective'' short paper published at ESOP 2024~\citep{LindleyMMSWY24}. The additions include the following:
\begin{itemize}
\item the scoped effect of explicit nondeterminism with cut treated as a parameterized theory~(\Cref{ex:cut-scope}) and its free model~(\Cref{sec:nond-with-cut});

\item a discussion about how monads support operations~(\Cref{def:op-on-monad,def:scoped-op-on-monad}), leading to a comparison between models of algebraic theories and models of parameterized theories in~\Cref{prop:pat-monad-on-set,prop:restricts-to-monad-for-alg-theory};

\item a new section about constructing a parameterized theory from an arbitrary scoped operation~(\Cref{sec:gener-param-theory}) and how their models are related~(\Cref{thm:scoped-theory-from-pat});
  
\item proofs for the general theorems about models from~\Cref{sec:free-models}, as well as proofs of freeness for the examples of models from~\Cref{sec:free-models-scoped}.
  The proofs of freeness involve exhibiting normal forms for each parameterized theory;


%
%
\item an expanded discussion of existing work about the higher-order syntax approach to scoped effects, in~\Cref{sec:scop-effects-handl};
  
\item the observation that one of the equations in the parameterized theory on nondeterminism with once is actually derivable from the others~(\Cref{exa:nondet-once-eq});

\item more details about the alternative definition of catching exceptions as a parameterized theory~(\Cref{rem:catch-flexi}); and

\item a discussion of future research about combinations of scoped theories in~\Cref{sec:summ-rese-direct}.  
  
\end{itemize}

\section{Background}

\subsection{Algebraic Effects}\label{sec:algebr-effects-handl}

Moggi \cite{Moggi89,Moggi91} shows that many non-pure features of programming languages, typically referred to as \emph{computational effects}, can be modelled uniformly as \emph{monads}, but the question is --- \emph{how do we construct a monad for an effect}, or putting it differently, \emph{where do the monads modelling effects come from}?
A classical result in category theory is that finitary monads over the category of sets are equivalent to \emph{algebraic theories} \cite{Linton66,Lawvere63}: an algebraic theory gives rise to a finitary monad by the free-algebra construction, and conversely every finitary monad is presented by a certain algebraic theory.
Motivated by this correspondence, Plotkin and Power~\cite{PlotkinP02} show that many monads that are used for modelling computational effects can be presented by algebraic theories of some basic effectful operations and some computationally natural equations.
This observation led them to the following influential perspective on computational effects \cite{PlotkinP02}, which is nowadays commonly referred to as \emph{algebraic effects}.

\begin{perspective}[\citet{PlotkinP02}]\label{persp:effect-as-alg-theory}
  An effect is realized by an algebraic theory of its basic operations, so it \emph{determines} a monad but is not identified with the monad.
\end{perspective}

We review the framework of algebraic effects in the simplest form here and
refer the reader to \citet{PP04Over} and \citet{bauer2019algebraic} for more
discussion.

\begin{definition}\label{def:alg:sigs}
  A (first-order finitary) \emph{algebraic signature} $\Sigma = \tuple{\ver{\Sigma}, \ar}$ consists of a set $\ver{\Sigma}$, whose elements are referred to as \emph{operations}, together with a mapping $\ar : \ver{\Sigma} \to \bN$, associating a natural number to each operation, called its \emph{arity}.
\end{definition}

Given a signature $\Sigma = \tuple{\ver{\Sigma}, \ar}$, we will write $O : n$ for an operation $O \in \ver{\Sigma}$ with $\ar(O) = n$.
The \emph{terms} $\Tm_\Sigma(\Gamma)$ in a context $\Gamma$, which is a finite list of variables, are inductively generated by
\begin{equation}\label{eq:alg:term:rules}
  \inferrule
  { }
  {\typing{\Gamma,x,\Gamma'}{x}}
  \hspace{2cm}
  \inferrule
  {O : n \\ \typing{\Gamma}{t_i} \text{ for } i = 1 \dots n}
  {\typing{\Gamma}{O(t_1, \dots, t_n)}}
\end{equation}
As usual we will consider terms up to renaming of variables.
Thus a context $\Gamma = (x_1, \dots, x_n)$ can be identified with the natural number $n$, and $\Tm_\Sigma$ can be thought of as a function $\bN \to \Set$.

\begin{example}\label{ex:sig:nondet}
  The signature of \emph{explicit nondeterminism} has two operations:
  \begin{mathpar}
     \orop : 2
     \and
     \fail : 0.
  \end{mathpar}
  Some small examples of terms of this signature are
  \begin{mathpar}
    \typing{}{\fail}
    \and
    \typing{x, y, z}{\orop(x, \orop(y, z))}
    \and
    \typing{x, y, z}{\orop(\orop(x, y), \fail)}
  \end{mathpar}
\end{example}

\begin{example}\label{ex:sig:state}
  The signature of \emph{mutable state} of a single bit has operations:
  \begin{mathpar}
    \putop{0} : 1
    \and
    \putop{1} : 1
    \and
    \get : 2.
  \end{mathpar}
  The informal intuition for a term $\typing{\Gamma}{\putop{i}(t)}$ is a program that writes the bit $i \in \brc{0, 1}$ to the mutable state and then continues as another program $t$, and a term $\typing{\Gamma}{\get(t_0, t_1)}$ is a program that reads the state, and continues as $t_i$ if the state is $i$.
  For example, the term $\typing{x, y}{\putop{0}(\get(x, y))}$ first writes $0$ to the state, then reads $0$ from the state, so always continues as $x$.
  For simplicity we consider a single bit, but multiple fixed locations and other storage are possible~\cite{PlotkinP02}.
\end{example}

\begin{definition}\label{def:alg:theories}
  A (first-order finitary) \emph{algebraic theory} $\MCT = \tuple{\Sigma, E}$ is a signature $\Sigma$ (\Cref{def:alg:sigs}) and a set $E$ of \emph{equations} of the signature $\Sigma$, where an equation is a pair of terms $\typing{\Gamma}{L}$ and $\typing{\Gamma}{R}$ under some context $\Gamma$.
  We will usually write an equation as $\typing{\Gamma}{L = R}$.
\end{definition}

\begin{example}
  The theory of \emph{exception throwing} has a signature containing a single operation $\throw : 0$ and no equations.
  The intuition for $\throw$ is that it throws an exception and the control flow never comes back, so it is a nullary operation.
\end{example}

\begin{example}\label{ex:nondet-eq}
  The theory of \emph{explicit nondeterminism} has the signature in \Cref{ex:sig:nondet} and the following equations saying that $\fail$ and $\orop$ form a monoid:
  \begin{mathpar}
    \typing{x}{\orop(\fail, x) = x}
    \quad
    \typing{x}{\orop(x, \fail) = x}
    \quad
    \typing{x, y, z}{\orop(x, \orop(y, z)) = \orop(\orop(x, y), z)}
  \end{mathpar}
    Following Plotkin and Pretnar~\cite{PlotPret13Hand} we refer to this as ``explicit'' non-determinism as it does not include full symmetry laws (that is, $\orop(x,y)=\orop(y,x)$)
    or idempotence laws ($\orop(x,x)=x$).

%
\end{example}

\begin{example}\label{sec:algebraic-effects}
  The theory of \emph{mutable state} of a single bit has the signature in \Cref{ex:sig:state} and the following equations for all $i, i' \in \brc{0, 1}$:
  \begin{gather}
    \typing{x_0, x_1}{\putop{i}(\get(x_0, x_1)) = \putop{i}(x_i)} \label{eq:putget} \\
    \typing{x}{\putop{i}(\putop{i'}(x)) = \putop{i'}(x)} \label{eq:putput}\\
    \typing{x}{\get(\putop{0}(x), \putop{1}(x)) = x} \label{eq:getput}
  \end{gather}
  The conspicuously missing equation for a $\get$ after a $\get$
  \[
    \typing{x_{00}, x_{01}, x_{10}, x_{11}}{\get(\get(x_{00}, x_{01}), \get(x_{10}, x_{11})) = \get(x_{00}, x_{11})}
  \]
  can be derived in the equational logic of algebraic theories, which will be
  introduced in a more general setting later in
  \Cref{sec:param:theories:scoped}, from the equations above: 
  \begin{align*}
       & \get(\get(x_{00}, x_{01}), \get(x_{10}, x_{11})) &&\text{via \cref{eq:getput}, calling this term by $t$}\\
    =\ & \get(\putop{0}(t), \putop{1}(t)) &&\text{via \cref{eq:putget} for $\putop{0}(t)$ and $\putop{1}(t)$} \\
    =\ & \get(\putop{0}(\get(x_{00}, x_{01}), \putop{1}(\get(x_{10}, x_{11}))) &&   \text{via \cref{eq:putget} again} \\
    =\ & \get(\putop{0}(x_{00}), \putop{1}(x_{11})) &&   \text{via \cref{eq:putget} right-to-left} \\
    =\ & \get(\putop{0}(\get(x_{00}, x_{11})), \putop{1}(\get(x_{00}, x_{11}))) &&   \text{via \cref{eq:getput}} \\
    =\ & \get(x_{00}, x_{11})
  \end{align*}
\hide{
\[
    {\putop{i}(\get(x_0, x_1)) = \putop{i}(x_i)}
    \quad
    {\putop{i}(\putop{i'}(x)) = \putop{i'}(x)}
    \quad
    {\get(\putop{0}(x), \putop{1}(x)) = x}
\]
}
\end{example}

\begin{remark}\label{rem:mon-alg-theory}
Every algebraic theory gives rise to a monad on the category $\Set$ of sets by the \emph{free-algebra construction}, which we will discuss in a more general setting in \Cref{sec:free-models}.
The three examples above respectively give rise to the monads $(1 + \blank)$, $\listf$, and $(\blank \times \Bit)^\Bit$ which are used to give semantics to the respective computational effects in programming languages \cite{Moggi89,Moggi91}.
\end{remark}

In this way, the monad for a computational effect is constructed in a remarkably intuitive manner, and this approach is highly composable: one can take the disjoint union of two algebraic theories to combine two effects, and possibly add more equations to characterise the interaction between the two theories \cite{HPP06Combine}.
By contrast, 
coproducts of monads, which correspond to taking the disjoint union of algebraic theories,
are much harder to describe explicitly even when they exist.

The kind of plain algebraic theory encapsulated by \Cref{def:alg:theories} above is not, however, sufficiently expressive enough for some programming language applications.
In this paper we focus on two problems with plain algebraic theories:

\begin{enumerate}
  \item Firstly, monadic bind for the monad generated by an algebraic theory is defined using \emph{simultaneous substitution of terms}:
    given a term $t \in \Tm_\Sigma(\Gamma)$ in a context $\Gamma$ and a mapping $\sigma : \Gamma \to \Tm_\Sigma(\Gamma')$ from variables in $\Gamma$ to terms in some context $\Gamma'$, the monadic bind $t \bind \sigma$ is defined to be
    the simultaneous substitution $t[\sigma]$ of $\sigma$ in $t$:
    \begin{mathpar}
      x[\sigma] = \sigma(x)
      \and
      O(t_1, \dots, t_n)[\sigma]  = O(t_1[\sigma], \dots, t_n [\sigma]).
    \end{mathpar}
    Monadic bind for a monad is used for interpreting \emph{sequential composition} of computations. 
    Therefore, the second clause above implies that every algebraic effect operation \emph{must} commute with sequential composition.
    However, in practice not every effectful operation enjoys this property.

  \item Secondly, it is common to have \emph{multiple instances} of a computational effect that can be dynamically created.
    For example, it is typical in practice to have an effectful operation $\openop$ that creates a ``file descriptor'' for a file at a given path, and for each file descriptor there is a pair of read and write operations that are independent of those for other files.

\end{enumerate}

These two restrictions have been studied separately, and different extensions to algebraic theories generalising \Cref{def:alg:theories} have been proposed for each:
\emph{scoped algebraic effects} for the first problem above and \emph{parameterized algebraic effects} for the second.
At first glance, the two problems seem unrelated, but the key insight of this paper is that scoped effects can be fruitfully understood as a \emph{non-commutative linear} variant of parameterized effects.

\subsection{Scoped Effects}\label{sec:scop-effects-handl}
The first problem with plain algebraic theories above is that operations must commute with sequential composition.
Therefore an operation $O(a_1, \dots, a_n)$ is ``atomic'' in the sense that it may not delimit a fresh \emph{scope}.
Alas, in practice it is not uncommon to have operations that do delimit scopes.
An example is \emph{exception catching}: $\catch(p, h)$ is a binary operation on computations that first tries the program $p$ and if $p$ throws an exception then $h$ is run.
The $\catch$ operation does not commute with sequential composition as $\catch(p, h) \bind f$ behaves differently from $\catch(p \bind f, h \bind f)$.
The former catches only the exceptions in $p$ whereas the latter catches exceptions both in $p$ and in $f$.
Further examples include operations such as opening a file in a scope, running a program concurrently in a scope, and looping a program in a scope.

Operations delimiting scopes are treated as \emph{handlers} (i.e.\ models) of algebraic operations by Plotkin and Pretnar~\cite{PlotPret13Hand}, instead of operations in their own right.
The following alternative perspective was first advocated by Wu et al.\ \cite{WSH2014}.

\begin{perspective}[\citet{WSH2014}]\label{persp:scoped-effects}
  Since sequential composition of monads generated from algebraic theories is \emph{substitution}, scoped operations are \emph{operations that do not commute with substitution}.
\end{perspective}

Operations that do not commute with substitution arise in contexts other than computational effects as well, for example, the later modality in guarded dependent type theory \cite{GDTT2016}.

We can use the results of Plotkin and Power~\cite{pp-alg-op-gen-eff} to give a precise phrasing of this on the side of semantics.
For simplicity we restrict this definition to the case of monads on $\Set$.
\begin{definition}[\citet{pp-alg-op-gen-eff}]\label{def:op-on-monad}
A monad $T$ on $\Set$ can be said to \emph{support an algebraic operation} $O:n$
if it is equipped with 
a natural transformation
\begin{displaymath}
  \widehat O_A : (TA)^n \to TA
\end{displaymath}
which moreover satisfies
\begin{equation}\mu_A \circ \widehat O_{TA} = \widehat O_A \circ (\mu_A)^n\text.\label{algebraicity}
\end{equation}
\end{definition}
  For example, any free monad $T$ determined by an algebraic theory $\MCT = \tuple{\Sigma,E}$
  according to Remark~\ref{rem:mon-alg-theory} supports all the operations in the signature $\Sigma$.
  As shown by \citet{pp-alg-op-gen-eff}, equation~\eqref{algebraicity} is another phrasing of commuting with sequential composition,
  and indeed to give such a support for a signature~$\Sigma$ is to simply give a family $(\Check O \in T(n)~|~(O:n)\in \Sigma)$ of ``generic effects'',
  for then we can write $\widehat O_{A}(\vec t)=\Check O  \bind \vec t$. 
Hence we can make a precise definition of scoped operations by dropping~\eqref{algebraicity}, following~\citet{YW2023}:

\begin{definition}[\citet{YW2023}]\label{def:scoped-op-on-monad}
  Let $T$ be a monad on $\Set$.
  A \emph{scoped operation on the monad $T$}, of arity $k\in\bN$, is a family of functions:
  \begin{displaymath}
    \sigma_A : (TA)^k \rightarrow TA
  \end{displaymath}
  subject only to naturality in $A$ (but not~\eqref{algebraicity}).
\end{definition}

Extensions of effect handlers to natively accommodate scoped operations were first studied by \citet{WSH2014} in Haskell, where the authors proposed two approaches:
\begin{enumerate}[parsep=1ex]
  \item\label{item:1} The \emph{bracketing approach} uses a pair of \emph{algebraic} operations $\beginop_s$ and $\opend_s$ to encode a scoped operation $s$.
    For example, consider the program ${s}(\putop{0}); \putop{1}; x$, expressed in direct style, 
    which writes the value $0$ in the scope $s$, then writes the value $1$ outside of that scope 
    and continues with the rest of the program $x$.  This can be encoded formally as:
    \[
      \beginop_s(\putop{0}(\opend_s(\putop{1}(x)))).
    \]


  \item\label{item:2} The \emph{higher-order syntax approach} directly constructs the syntax
    for programs with algebraic and scoped operations as an initial algebra
    of an endofunctor over the category of (finitary) endofunctors over $\Set$.
    Elementarily, this amounts to adding the following rule,
    for every scoped operation $S$ that delimits $n$ scopes,
    to the rules in \cref{eq:alg:term:rules} generating terms:
    \begin{equation*}
      \inferrule
      {S : n \\
      \typing{x_1 , \dots , x_m}{t_i} \text{ for } i = 1 \dots n \\ \typing{\Gamma}{k_j} \text{ for } j = 1 \dots m }
           {\typing{\Gamma}{S(t_1 , \dots,  t_n)\brc{k_1, \dots, k_m}}}
    \end{equation*}
    The intuition is that $S(t_1, \dots, t_n)\brc{k_1, \dots, k_m}$
    is the term applying the scoped operation $S$ to the terms $t_1, \dots, t_n$
    followed by a \emph{formal substitution} that replaces the $m$ variables 
    in $t_i$ with the terms $k_1, \dots, k_m$ respectively (c.f.\ explicit
    substitution \cite{Ghani_Uustalu_Hamana_2006} and delayed substitution
    \cite{GDTT2016}).
    Since substitutions get stuck at scoped operations,
    they need to be kept in terms, if we want to have a monad of terms with
    scoped operations.

    Moreover, to model the intuition that 
    the number of internal variables $m$ 
    after a ``substitution'' should not be visible in the result of ``substitution''
    $S(t_1 , \dots,  t_n)\brc{k_1 , \dots, k_m}$,
    the terms are further quotiented by (the congruence relation generated by)
    the following rule: for all $m, m' \in \bN$ and functions $f : \brc{1 \dots m} \to \brc{1 \dots m'}$,
    \begin{equation*}
      \inferrule{ 
        S : n \\
        \typing{x_1, \dots, x_m}{t_i} \text { for } i = 1 \dots n \\
        \typing{\Gamma}{k_{j'} \text{ for  } j' = 1  \dots m'}
      }
      { \typing{\Gamma}{S(t_1, \dots, t_n)\brc{ k_{f(1)}, \dots, k_{f(m)}}   
            =  S(t_1[f], \dots, t_n[f])\brc{ k_1, \dots, k_{m'}}} }
    \end{equation*}
    where the terms $\typing{x_1 \dots x_{m'}}{t_i[f]}$ are obtained by
    replacing each variable $x_j$ in $t_i$ with $x_{f(j)}$,  for $1 \leq j \leq
    m$.
    This rule can be motivated as follows: in the left-hand side, 
every variable $x_j$ in each $t_i$ is formally replaced by $k_{f(j)}$ according to
    the formal substitution,
    while in the right-hand side, $x_j$ is first replaced by $x_{f(j)}$
    by the true substitution $[f]$, and then formally replaced by $k_{f(j)}$ according
    to the formal substitution.
    These two results should be exactly the same since we would like our formal substitution
    to behave like true substitutions.

    Terms in a context $\Gamma$ obtained in this way also form a monad~\cite{YPWvS2021}.
  \end{enumerate}
  \todo[inline]{C: I think we agreed in the meeting to replace the code above with something else.}
The higher-order syntax approach was regarded the more principled one since in the first approach ill bracketed pairs of $\beginop_s$ and $\opend_s$ are possible, such as
\[
  \opend_s(\putop{0}(\beginop_s(\beginop_s(\putop{1}(x))))).
\]
In subsequent work, both of these two approaches received further development \cite{DBLP:conf/lics/PirogSWJ18,YW2023,Berg2023framework,YPWvS2021} and operational semantics for scoped effects has also been developed \cite{Bosman2023calculus}.
Of particular relevance to the current paper is the work of \citet{DBLP:conf/lics/PirogSWJ18}, which we briefly review in the rest of this section.

\subsubsection*{Related work on models for scoped effects}

\citet{DBLP:conf/lics/PirogSWJ18} fix the ill-bracketing problem in the bracketing approach by considering the category $\Set^\bN$ whose objects are sequences $X = (X(0),X(1),\ldots)$ of sets and morphisms are just sequences of functions.
Given $X \in \Set^\bN$, the idea is that $X(n)$ represents a set of terms at \emph{bracketing level} $n$ for every $n \in \bN$.

On this category, there are two functors $(\later), (\earlier) : \Set^\bN \to \Set^\bN$, pronounced ``later'' and ``earlier'', that shift the bracketing levels:
\begin{equation}\label{eqn:later-earlier}
  (\later X)(0) = \emptyset,               \quad\quad
  (\later X) (n + 1) = X(n),       \quad\quad
  (\earlier X)(n) = X (n + 1).
\end{equation}
These two functors are closely related to the bracketing approach~(\ref{item:1}):
a morphism $b : \earlier X \to X$ for a functor $X$ opens a scope, turning a term $t$ at level $n+1$ to the term $\beginop(t)$ at level $n$.
Conversely, a morphism $e : \later X \to X$ closes a scope, turning a term $t$ outside the scope, so at level $n-1$, to the term $\opend(t)$ at level $n$.

Given two signatures $\Sigma$ and $\Sigma'$ as in \Cref{def:alg:sigs} for algebraic and scoped operations respectively, let $\bar{\Sigma}, \bar{\Sigma'} : \Set^\bN \to \Set^\bN$ be the functors given by
\[\textstyle
  (\bar{\Sigma} X)(n) = \coprod_{o \in \ver{\Sigma}} X(n)^{\ar(o)}
  \quad\text{and}\quad
  (\bar{\Sigma'} X)(n) = \coprod_{s \in \ver{\Sigma'}} X(n)^{\ar(s)}.
\]
Moreover, for every $A \in \Set$, let $\upp{A} \in \Set^\bN$ be given by
\begin{mathpar}
  (\upp{A})(0) = A  \and
  (\upp{A})(n+1) = 0,
\end{mathpar}
and conversely for every $X \in \Set^\bN$, let $\downn X \in \Set$ be given by $\downn X = X(0)$.
These are actually adjoint functors $\upp{} \dashv \downn{}$~.
Now~\Cref{prop:pswj-monads-isomorphic} constructs the syntactic monad for programs with the given algebraic and scoped operations, without taking into account any equations.



\begin{proposition}[\citet{DBLP:conf/lics/PirogSWJ18}]\label{prop:pswj-monads-isomorphic}
  The following functor can be extended to a monad that is isomorphic to the monad obtained by the higher-order syntax approach (\ref{item:2}):
\[
  \downn{\ } \circ \paren{ \bar{\Sigma} + \paren{ \bar{\Sigma'} \circ \earlier } + \later }^* \circ \upp{}
  : \Set \to \Set
\]
where $(\blank)^*$ is the free monad over an endofunctor.
\end{proposition}

\citet{DBLP:conf/lics/PirogSWJ18} defines a \emph{model} of a scoped effect as an \emph{algebra} for the monad $\paren{ \bar{\Sigma} + \paren{ \bar{\Sigma'} \circ \earlier } + \later }^*$ on $\Set^\bN$.
This is a notion of handler accommodating both algebraic and scoped effects.
In~\Cref{thm:noeqn} we show that this monad on $\Set^\bN$ arises as a special case of the free monad for a parameterized algebraic theory.
In Thms.~\ref{thm:nondet-free-model}--\ref{thm:state-free-model}, we show that three examples of models of \citet{DBLP:conf/lics/PirogSWJ18} are actually \emph{free algebras} on $\upp A\in\Set^\bN$ for an appropriate set of equations for each example.

\subsection{Parameterized Algebraic Theories}
\label{sec:pat}

Recall that our second problem with plain algebraic theories is that they do not support the dynamic creation of multiple instances of computational effects.
This problem, sometimes known as the \emph{local computational effects} problem, was first systematically studied by \citet{Power2006} in a purely categorical setting.
A syntactic framework extending that of algebraic theories, called \emph{parameterized algebraic theories}, was introduced by Staton~\cite{DBLP:conf/fossacs/Staton13,DBLP:conf/lics/Staton13} and is used to give an axiomatic account of local computational effects such as restriction~\cite{Pitts2011}, local state~\cite{PlotkinP02}, and the $\pi$-calculus~\cite{Milner1999,DBLP:journals/tcs/Stark08}.

Operations in a parameterized theory are more general than those in an algebraic theory because they may \emph{use} and \emph{create} values in an \emph{abstract} type of parameters.
The parameter type has different intended meanings for different examples of parameterized theories, typically as some kind of resource such as memory locations or communication channels.
In this paper, we propose to interpret parameters as \emph{names of scopes}.
\begin{perspective}
  Scoped operations can be understood as operations allocating and consuming instances of a resource: the names of scopes.
\end{perspective}

In the case of local state, the operations of \Cref{ex:sig:state} become $\get(a, x_0, x_1)$ and $\putop{i}(a,x)$, now taking a parameter $a$ which is the location being read or written to.
In a sense, each memory location $a$ represents an \emph{instance} of the state effect, with its own $\get$ and $\putop{}$ operations.
We also have a term $\newop{i}(a.x(a))$ which allocates a fresh location named $a$ storing an initial value $i$, then continues as $x$; the computation $x$ might mention location $a$.
The following is a possible equation, which says that reading immediately after allocating is redundant:
\begin{displaymath}
  \newop{i}(a.\get(a,x_0(a),x_1(a))) = \newop{i}(a.x_i(a)).
\end{displaymath}
For the full axiomatization of local state see~\cite[\S V.E]{DBLP:conf/lics/Staton13}.
A closed term can only mention locations introduced by $\newop{i}$, meaning that type of locations is \emph{abstract}.

To model scoped operations, we think of them as allocating a new scope.
For example, the scoped operation $\once$, which chooses the first non-failing branch of a nondeterministic computation, is written as $\once(a.x(a))$.
It creates a new scope $a$ and proceeds as $x$.
As in~\S\ref{sec:introduction}, there is an explicit operation $\close(a,x)$ for closing the scope $a$ and continuing as $x$.

Well-formed programs close scopes precisely once and in the reverse order to their allocation.
Thus in~\S\ref{sec:param:theories:scoped} we will discuss a \emph{non-commutative linear} variation of parameterized algebraic theories needed to model scoped effects.
With our framework we then give axiomatizations for examples from the scoped effects literature (Thms.~\ref{thm:nondet-free-model}--\ref{thm:state-free-model}).

Our parameters are linear in the same sense as variables in linear logic and linear lambda calculi~\cite{DBLP:journals/tcs/Girard87,benton-wadler-linear-logic-monads-and-the-lambda-calculus}, but with an additional non-commutativity restriction.
Non-commutative linear systems are also known as ordered linear systems~\cite{10.5555/935673,DBLP:conf/popl/PetersenHCP03}.
A commutative linear version of parameterized algebraic theories was considered by \citet{DBLP:conf/popl/Staton15} to give an algebraic theory of quantum computation; in this case, parameters stand for qubits.

\begin{remark}
  Parameterized algebraic theories characterize a certain class of enriched monads~\cite{DBLP:conf/fossacs/Staton13}, extending the correspondence between algebraic theories and monads on the category of sets, and the idea of Plotkin and Power~\cite{PlotkinP02} that computational effects give rise to monads (see~\S\ref{sec:algebr-effects-handl}).
Thus, the syntactic framework of parameterized theories has a canonical semantic status.
We can use the monad arising from a parameterized theory to give semantics to a programming language containing the effects in question.
\end{remark}

The framework of parameterized algebraic theories is related to graded theories~\cite{10.1145/3547654}, which also use presheaf-enrichment; second-order algebra~\cite{DBLP:journals/pacmpl/FioreS22,DBLP:conf/csl/FioreH10,DBLP:conf/mfcs/FioreM10}, which also use variable binding; and graphical methods~\cite{10.1007/978-3-319-08918-8_23}, which also connect to presheaf categories.


\section{Parameterized theories of scoped effects}\label{sec:param:theories:scoped}

In order to describe scoped effects we use a substructural version of parameterized algebraic theories~\cite{DBLP:conf/fossacs/Staton13}.
A theory consists of a signature (Def.~\ref{def:arity-signature}) and equations (Def.~\ref{def:presentation}) between terms formed from the signature.
Terms contain two kinds of variables: computation variables ($x$, $y$, \dots), which each expect a certain number of parameters, and parameter variables ($a$, $b$, \dots).
In the case of scoped effects, a parameter represents the name of a scope.

Any scoped signature gives rise to a parameterized algebraic signature (Section~\ref{sec:pat:def}). But a parameterized algebraic \emph{theory} includes equations as well as operations. In Section~\ref{sec:pat:ex} we propose equational theories for various scoped effects: exceptions, state, and non-determinism. In Section~\ref{sec:free-models-scoped} we show that the free models of these theories completely capture 
these notions from the literature. 

\subsection{Parameterized Algebraic Theories with Non-Commutative Linear Parameters}
\label{sec:pat:def}

\begin{definition}\label{def:arity-signature}

  A \emph{(parameterized) signature $\Sigma = \tuple{\ver{\Sigma}, \ar}$} consists of a set of operations $\ver{\Sigma}$ and for each operation $\op \in \ver{\Sigma}$ a \emph{parameterized arity} $\ar(\op) = {(p\vbar m_1\sdots m_k)}$ consisting of a natural number $p$ and a list of natural numbers $m_1,\sdots,m_k$.
  This means that the operation $\op$ takes in $p$ parameters and $k$ continuations, and it binds $m_i$ parameters in the $i$-th continuation.
\end{definition}

\begin{remark}\label{rem:scoped-signature-to-parameterized-signature}
Given signatures for algebraic and scoped operations, as in Def.~\ref{def:alg:sigs} and~\S\ref{sec:scop-effects-handl}, we can translate them to a parameterized signature as follows:
\begin{itemize}
\item for each algebraic operation $(\mathsf{op}:k)$ of arity $k \in \bN$, there is a parameterized operation with arity $(0 \vbar 0\sdots 0)$, where the list $0\sdots0$ has length $k$;
\item for each scoped operation $(\mathsf{sc}:k)$ of arity $k \in \bN$, there is a parameterized operation $\mathsf{sc} : (0 \vbar 1\sdots 1)$, where the list $1\sdots1$ has length $k$;
\item there is an operation $\close : (1 \vbar 0)$, which closes the most recent scope, and which all the different scoped operations share.
\end{itemize}
Not every parameterized signature arises from a scoped signature, because in general we may have operations where the arguments have different valences (e.g.~$m_1\neq m_2$). So parameterized signatures give some more flexibility. We explore this point more in Remark~\ref{rem:scoped-signature-to-parameterized-signature}.
\end{remark}


\begin{example}\label{ex:nondet-once-sig}
  The algebraic theory of explicit nondeterminism in \Cref{ex:sig:nondet} can be extended with a \emph{semi-determinism} operator $\once$:
 \begin{mathpar}
\orop : (0\vbar 0,0)  \and \once:(0 \vbar 1) \and \fail:(0\vbar -) \and \close:(1 \vbar 0)
\end{mathpar}
(We write $-$ for an empty list or context.)
The continuation of $\once$ opens a new scope, by binding a parameter.
Inside this scope, only the first successful branch of $\orop$ is kept.

We consider the terms and equations in Example~\ref{exa:nondet-once-eq} and Fig.~\ref{fig:nondet-once}.

\end{example}

\begin{figure}[t]
  \begin{align}
    \label{eq:3}  x:0,y:0,z:0\vbar - &\vdash \orop(\orop(x,y),z) = \orop(x,\orop(y,z)) \\
    x:0 \vbar - &\vdash \orop(x,\fail) = x     \label{eq:3a} \\
     x:0 \vbar - &\vdash \orop(\fail,x) = x\label{eq:3b}\\
    \label{eq:4a} - \vbar - &\vdash \once(a.\fail) = \fail\\
    x:1 \vbar - &\vdash \once(a.\orop(x(a),\,x(a))) = \once(a.x(a)) \label{eq:4b}\\
    x:0 \vbar - &\vdash \once(a.\close(a,x)) = x
  \label{eq:4c}
    \\
  \label{eq:4}
  x :0,y : 1 \vbar - &\vdash \once\bigl(a.\orop(\close(a,x),\, y(a))\bigr) = x
  \end{align}
\caption{The parameterized theory of explicit nondeterminism (\ref{eq:3}--\ref{eq:3b}) and $\once$
(\ref{eq:4a}--\ref{eq:4}). \label{fig:nondet-once}}
\end{figure}

For a given signature, we define the terms-in-context of \emph{algebra with non-commutative linear parameters}.
A context $\Gamma$ of \emph{computation variables} is a finite list $x_1 : p_1, \dots, x_n : p_n$, where each variable $x_i$ is annotated with the number $p_i$ of parameters it consumes.
A context $\Delta$ of \emph{parameter variables} is a finite list $a_1, \dots, a_m$.
Terms $\typing{\Gamma \vbar \Delta}{t}$
are inductively generated as follows:
\begin{gather}
  \label{eq:var-former}
  \inferrule{ }{\Gamma, x:p, \Gamma' \vbar  a_1\sdots a_p \vdash x(a_1\sdots a_p)} \\[0.5em]
  \label{eq:op-former}
  \inferrule
  {\Gamma \vbar \Delta, b_1\sdots b_{m_1}\vdash t_1 \\
    \dots \\ \Gamma \vbar \Delta, b_1\sdots b_{m_k}\vdash t_k \\
    \op :(p\vbar m_1 \sdots m_k)}
  {\Gamma \vbar \Delta, a_1\sdots a_p \vdash \op(a_1\sdots a_p,\; b_1\sdots b_{m_1}.t_1\ \dots\ b_1\sdots b_{m_k}.t_k)} 
\end{gather}
In the conclusion of the last rule, the parameters $a_1\sdots a_p$ are consumed by the operation $\op$.
The parameters $b_1\sdots b_{m_i}$ are bound in $t_i$.
As usual, we treat all terms up to renaming of variables.
When $p=k=0$, i.e.~$\op$ is a constant symbol, we omit the parentheses.

The context $\Gamma$ of computation variables admits the usual structural rules: weakening, contraction, and exchange; the context $\Delta$ of parameters does not.
All parameters in $\Delta$ must be used exactly once, in the reverse of the order in which they appear.
Intuitively, a parameter in $\Delta$ is the name of an open scope, so the restrictions on $\Delta$ mean that scopes must be closed in the opposite order that they were opened, that is, scopes are well-bracketed.
The arguments $t_1,\sdots,t_k$ of an operation $\op$ are continuations, each corresponding to a different branch of the computation, hence they share the parameter context $\Delta$. 

Compared to the \emph{algebra with linear parameters} used for describing quantum computation~\cite{DBLP:conf/popl/Staton15}, our syntactic framework has the additional constraint that $\Delta$ cannot be reordered.
Given these constraints, the context $\Delta$ is in fact a stack, so inside a term it is unnecessary to refer to the variables in $\Delta$ by name.
We have chosen to do so anyway in order to make more clear the connection to non-linear parameterized theories~\cite{DBLP:conf/fossacs/Staton13,DBLP:conf/lics/Staton13}.

The syntax admits the following simultaneous substitution rule:
  \begin{equation}\label{eq:general-subst-rule}
    \inferrule{(x_1:m_1\sdots x_l:m_l) \vbar \Delta \vdash t \\ \Gamma'\vbar \Delta',a_1\sdots a_{m_1} \vdash t_1\and \dots\and \Gamma'\vbar \Delta',a_1\sdots a_{m_l} \vdash t_l}
    {\Gamma' \vbar \Delta',\Delta \vdash t\bigl[ (\Delta',a_1\sdots a_{m_1}\vdash t_1)/x_1 \ \ldots \ (\Delta',a_1\sdots a_{m_l}\vdash t_l)/x_l \bigr]}
  \end{equation}
  In the conclusion, the notation $(\Delta',a_1\sdots a_{m_i}\vdash t_i)/x_i$ emphasizes that the parameters $(a_1\ldots a_{m_i})$ in $t_i$ are replaced by the corresponding parameters that $x_i$ consumes in $t$, either bound parameters or free parameters from $\Delta$.
  To ensure that the term in the conclusion is well-formed, we must substitute a term that depends on $\Delta'$ for \emph{all} the computation variables in the context of $t$.
  
  An important special case of the substitution rule is where we add a number of extra parameter variables to the beginning of the parameter context, increasing the sort of each computation variable by the same number.
  The following example instance of \eqref{eq:general-subst-rule}, where $\ar(\op) = (1 \vbar 1)$, illustrates such a ``weakening'' by adding two extra parameter variables $a'_1, a'_2$ and replacing $x : 2$ by $x' : 4$.
  \begin{displaymath}
    \inferrule{
      x:2 \vbar a_1,a_2 \vdash \op(a_2,b.x(a_1,b)) \\ x':4 \vbar a'_1,a'_2,b_1,b_2 \vdash x'(a'_1,a'_2,b_1,b_2)
    }{
      x':4\vbar a'_1,a'_2,a_1,a_2 \vdash \op(a_2,b.x'(a'_1,a'_2,a_1,b))
    }
  \end{displaymath}

\begin{definition}\label{def:presentation}
  An \emph{algebraic theory $\MCT = \tuple{\Sigma,E}$ with non-commutative linear parameters} is a parameterized signature $\Sigma$ together with a set $E$ of \emph{equations}.
  An equation is a pair of terms $\Gamma\vbar \Delta\vdash L$ and $\Gamma\vbar \Delta\vdash R$ in the same context $(\Gamma\vbar\Delta)$.
  We write an equation as $\Gamma\vbar \Delta\vdash L=R$.
\end{definition}

We will omit the qualifier ``with non-commutative linear parameters'' where convenient and refer to ``parameterized theories'' or just ``theories''.

Given a theory $\MCT=\tuple{\Sigma,E}$, we form a system of equivalence relations $\eqth \MCT \Gamma \Delta$ on terms in each context $(\Gamma\vbar\Delta)$, by closing \emph{substitution instances} of the equations in $E$ under reflexivity, symmetry, transitivity, and congruence rules:
\begin{mathpar}
  \inferrule{ }{\Gamma\vbar\Delta \vdash L = L} \and
  \inferrule{\Gamma\vbar\Delta \vdash L = R}{\Gamma\vbar\Delta \vdash R = L} \and
  \inferrule
  {\Gamma\vbar\Delta \vdash L = R \and \Gamma\vbar\Delta \vdash R = R'}
  {\Gamma\vbar\Delta \vdash L = R'} \and
  \inferrule
  {\Gamma \vbar \Delta, b_1\sdots b_{m_1}\vdash L_1=R_1 \\
    \dots \\ \Gamma \vbar \Delta, b_1\sdots b_{m_k}\vdash L_k=R_k \\
    \op :(p\vbar m_1 \sdots m_k) \in \Sigma}
  {\Gamma \vbar \Delta, a_1\sdots a_p \vdash \op(a_1\sdots a_p,\; b_1\sdots b_{m_1}.L_1\ \dots\ b_1\sdots b_{m_k}.L_k) = \op(a_1\sdots a_p,\; b_1\sdots b_{m_1}.R_1\ \dots\ b_1\sdots b_{m_k}.R_k)} 
\end{mathpar}


\subsection{Examples of Equations for Scoped Theories via Parameterized Presentations}
\label{sec:pat:ex}

\begin{example}\label{exa:nondet-once-eq}
We continue Example~\ref{ex:nondet-once-sig}, semi-determinism, with $\orop$, $\once$, $\fail$ and $\close$.
  The term formation rules in~\Cref{sec:pat:def} allow the most recently opened scope to be closed using the $\close$ operation by consuming the most recently bound parameter; $\close$ has one continuation which does not depend on any parameters.
  
  Examples of equations, i.e.~pairs of terms in the same context, are given in Figure~\ref{fig:nondet-once}.
  As an illustration, we note that Equation~(\ref{eq:4c}) is derivable from the others:
  \begin{align*}
    \once(a.\close(a,x)) &=
                           \once(a.\orop(\close(a,x),\fail))&&\text{via~(\ref{eq:3a})}
    \\&=
    x&&\text{via~(\ref{eq:4}).}
  \end{align*}
\end{example}
  
\begin{example}\label{exa:catch-signature}
  As we mentioned earlier, \emph{exception catching} is not an ordinary algebraic operation.
  The signature for throwing and catching exceptions is the following:
  \begin{mathpar}
  \throw:(0 \vbar -) \and \catch:(0 \vbar 1,1) \and \close:(1 \vbar 0)
\end{mathpar}
The $\throw$ operation uses no parameters and takes no continuations.
The $\catch$ operation uses no parameters and takes two continuations which each open a new scope, by binding a fresh parameter.
Exceptions are caught in the first continuation, and are handled using the second continuation.

The $\close$ operation uses one parameter and takes one continuation binding no parameters.
The term $\close(a, x)$ closes the scope named by $a$ and continues as~$x$.
For example, in the term $\catch(a.\close(a, x), b.y(b))$, exceptions in $x$ will not be caught, because the scope of the $\catch$ operation has already been closed.
The equations are:
\begin{align}
  \label{eq:8}
  y{:}0 \vbar - &\vdash \catch(a.\throw,\,b.\close(b,y)) = y \\
  -\vbar - &\vdash \catch(a.\throw,\,b.\throw) = \throw \label{eq:10}
\\
  x{:}0,y{:}1 \vbar - &\vdash \catch\bigl(a.\close(a,x),\,b.y(b)\bigr) = x
  \label{eq:9}
\end{align}
\end{example}



\begin{remark}\label{rem:catch-flexi}
  The arity of $\catch$ from Ex.~\ref{exa:catch-signature} corresponds to the signature used in~\cite[Example 4.5]{DBLP:conf/lics/PirogSWJ18}.
  Using the extra flexibility of parameterized algebraic theories, we could instead consider the arity $\varcatch:(0\vbar 1,0)$.
  Equations~(\ref{eq:8})--(\ref{eq:9}) become
\begin{align*}
  y{:}0 \vbar - &\vdash \varcatch(a.\throw,\,y) = y
  \\
  x{:}0,y{:}1 \vbar - &\vdash \varcatch\bigl(a.\close(a,x),\,y\bigr) = x
\end{align*}
This seems more natural as there is no need to delimit a scope in the second continuation, which handles the exceptions.

We could also encode $\varcatch$ in the theory of $\catch$ by making the following definition:
\begin{displaymath}
  \varcatch(a.x(a),y)=\catch(a.x(a),b.\close(b,y))
\end{displaymath}
and deducing the two equations for $\varcatch$ from~\cref{eq:8,eq:9}.
\end{remark}

\begin{example}[Mutable state with local values]\label{ex:state-local-sig}
  The theory of (boolean) mutable state with one memory location (Ex.~\ref{ex:sig:state}) can be extended with scoped operations $\local{0}$ and $\local{1}$ that write respectively $0$ and $1$ to the state.
  Inside the scope of $\local{}$, the value of the state just before the local is not accessible anymore, but when the $\local{}$ is closed the state reverts to this previous value.
  \begin{mathpar}
  \local{i} :(0 \vbar 1) \and \putop{i}:(0\vbar 0) \and \get:(0\vbar 0,0) \and \close:(1\vbar 0)
\end{mathpar}
The equations for the parameterized theory of state with $\local{}$ comprise the usual equations for state in the literature \cite{PlotkinP02,5570883}:
\begin{align}
  \label{eq:5}
  z:0 \vbar - &\vdash \get(\putop{0}(z),\putop{1}(z)) = z \\
  \label{eq:5a}
  z:0 \vbar - &\vdash \putop{i}(\putop{j}(z)) = \putop{j}(z) \\
  x_0:0,x_1:0 \vbar - &\vdash \putop{i}(\get(x_0,x_1)) = \putop{i}(x_i) 
\end{align}
together with equations for local/close, and the interaction with state:
\begin{align}
  x:0 \vbar - &\vdash \local{i}(a.\close (a,x)) = x \\
  x_0:1,x_1:1 \vbar - &\vdash \local{i}(a.\get(x_0(a),\, x_1(a))) = \local{i}(a.x_i(a)) \\
  \label{eq:6}
  z:1 \vbar - &\vdash \local{i}(a.\putop{j}(z(a))) = \local{j}(a.z(a))\\
  \label{eq:7}
  z:0 \vbar a &\vdash \putop{i}(\close(a,z)) = \close(a,z)
\end{align}
This extension of mutable state is different from the one discussed in~\S\ref{sec:pat}, where memory locations can be dynamically created.
\end{example}


\begin{example}[Explicit nondeterminism with cut]\label{ex:cut-scope}
  The theory of explicit nondeterminism given by~\crefrange{eq:3}{eq:3b} can be extended with an operation that prunes the list of possible results, similar to cut in Prolog:
  \begin{mathpar}
    \cut : (0\vbar 0) \and \orop : (0\vbar 0,0)  \and \fail:(0\vbar -)
  \end{mathpar}
  The $\cut$ operation is algebraic and has one continuation; $\cut(x)$ intuitively discards the choices that have not been explored yet, and returns all the possible results of $x$.
  The behaviour of $\cut$ has been axiomatized in~\cite[Sec.~6]{DBLP:journals/jfp/PirogS17}:
  \begin{align}
    \label{eq:15}
    x:0,y:0 \vbar - &\vdash \orop(\cut(x),y) = \cut(x) \\
    \label{eq:14}
    x:0,y:0 \vbar - &\vdash \orop(x,\cut(y)) = \cut(\orop(x,y)) \\
    \label{eq:13}
  x:0 \vbar - &\vdash \cut(\cut(x)) = \cut(x)
  \end{align}

  To delimit the scope in which $\cut$ discards choices, we add an operation that opens a scope, and a $\close$ operation for closing scopes:
  \begin{mathpar}
    \scope : (0 \vbar 1) \and \close: (1 \vbar 0)
  \end{mathpar}
  We propose the following equations for $\scope$:
  \begin{align}
    \label{eq:16}
    -\vbar - &\vdash \scope(a.\fail) = \fail \\
    \label{eq:11}
    x:0 \vbar - &\vdash \scope(a.\cut(x(a))) = \scope(a.x(a)) \\
    \label{eq:12}
    x:0,y:1 \vbar - &\vdash \scope(a.\orop(\close(a,x),\,y(a))) = \orop(x,\,\scope(a.y(a)))                      
\end{align}
Intuitively,~\cref{eq:11} says that when a $\cut$ reaches the boundary of a scope, the $\cut$ is erased so it cannot affect choices outside of the scope.
The Haskell implementation of the $\mathtt{scope}$ function by
\citet{DBLP:journals/jfp/PirogS17} has similar behaviour.

\Cref{eq:12} axiomatizes the interaction between opening and closing a scope.
From it, we can derive the following:
\begin{align*}
  \scope(a.\close(a,x)) &= \scope(a.\orop(\close(a,x),\fail)) &&\text{via~(\ref{eq:3a})}\\
                        &= \orop(x,\scope(a.\fail)) &&\text{via~(\ref{eq:12})} \\
                        &= \orop(x,\fail) &&\text{via~(\ref{eq:16})} \\
                        &= x &&\text{via~(\ref{eq:3a})}
\end{align*}

\end{example}

\begin{remark}
  We can almost encode the theory of $\once$ from~\Cref{exa:nondet-once-eq} into the theory for $\cut$ and $\scope$~(\Cref{ex:cut-scope}), by defining $\once$ to be $\scope$, and defining the $\close(a,x)$ of the $\once$ theory to be $\cut(\close(a,x))$.
  Then we can recover~\cref{eq:4a,eq:4c,eq:4} from \Cref{fig:nondet-once}, by using~\crefrange{eq:15}{eq:12}; but we cannot recover~\cref{eq:4b}, which is idempotence of $\orop$ inside a $\once$.
\end{remark}

\section{Models of parameterized theories}










\subsection{Models in $\Set^\bN$}\label{sec:models-Set-bN}

Models of first-order algebraic theories~\cite{bauer2019algebraic} consist simply of a set together with specified interpretations of the operations of the signature, validating a (possibly empty) equational specification.
The more complex arities and judgement forms of a parameterized theory require a correspondingly more complex notion of model.
Rather than simply being a set of abstract computations, a model will now be stratified into a sequence of sets $X = (X(0),X(1),\ldots) \in \Set^\bN$ where $X(n)$ represents computations \emph{taking $n$ parameters}.
In~\S\ref{sec:scop-effects-handl} we described Pir\'{o}g et al's somewhat different use of $\Set^\bN$~\cite{DBLP:conf/lics/PirogSWJ18}.
We connect the two approaches in Thm.~\ref{thm:noeqn} below.

At first glance, a term $x_1 : m_1, \ldots, x_k : m_k \vbar a_1, \ldots, a_p \vdash t$ should denote a function $X(m_1) \times \ldots \times X(m_k) \to X(p)$, since a $k$-tuple of possible continuations that consume different numbers of parameters is mapped to a computation that consumes $p$ parameters.
However, the admissible substitution rule~\eqref{eq:general-subst-rule} shows us that actually such a term must also denote a sequence of functions
\begin{displaymath}
  \den{x_1 : m_1, \sdots, x_k : m_k \vbar a_1, \sdots, a_p \vdash t}_{\MCX,n} : X(n + m_1) \times \ldots \times X(n + m_k) \to X(n + p).
\end{displaymath}

\begin{definition}\label{def:sigma-struct}
  Let $\Sigma$ be a parameterized signature (Def.~\ref{def:arity-signature}).
  A \emph{$\Sigma$-structure} $\MCX$ is an $X \in \Set^{\bN}$ equipped with, for each $\op : (p\vbar m_1\sdots m_k)$ and $n \in \bN$, a function
  \begin{displaymath}
    \op_{\MCX,n} : X(n + m_1) \times \ldots \times X(n + m_k) \to X(n + p).
  \end{displaymath}
\end{definition}

The interpretation of terms is now defined by structural recursion in a standard way, where the interpretation of a computation variable term such as ${x_1 : m_1, \ldots, x_k : m_k \mid a_1, \ldots, a_{m_i} \vdash x_i(a_1,\ldots,a_{m_i})}$ is given by the sequence of product projections
\begin{displaymath}
  X(n + m_1) \times \ldots \times X(n + m_i) \times \ldots \times X(n + m_k)
  \to X(n + m_i).
\end{displaymath}

\begin{definition}\label{def:model}
  Let $\MCT$ be a parameterized theory over the signature $\Sigma$.
  A $\Sigma$-structure $\MCX$ is a \emph{model} of $\MCT$ if for every equation $\Gamma\vbar\Delta \vdash s = t$ in $\MCT$, and every $n \in \bN$, we have an equality of functions $\den{\Gamma\vbar\Delta\vdash s}_{\MCX,n} = \den{\Gamma\vbar\Delta\vdash t}_{\MCX,n}$.
\end{definition}

\begin{proposition}
  The derivable equality $(=_\MCT)$ in a parameterized algebraic theory~$\MCT$ is sound:
every $\MCT$-model satisfies every equation of $=_\MCT$.
\label{thm:sound}
\end{proposition}
\begin{proof}[Proof notes]
  By induction on the structure of derivations.
\end{proof}
\newcommand{\power}[2]{[#1,#2]}
\begin{remark}\label{rem:models-via-enriched-categories}
  A more abstract view on models is based on enriched categories, since parameterized algebraic theories can be understood in terms of enriched Lawvere theories~\cite{enriched-lawvere,KELLY1993163,DBLP:conf/fossacs/Staton13}. This is useful because, by interpreting algebraic theories in different categories, we can combine the algebra structure with other structure, such as topological or order structure for recursion~\cite[\S
  6]{abramsky-jung}, or make connections with syntactic categories~\cite{STATON2014197}. Recall that the category $\Set^\bN$ has a `Day convolution' monoidal structure~\cite{10.1007/BFb0060438}:
  ${(X\otimes Y)(n) = \sum_{m_1+m_2=n} X(m_1)\times Y(m_2)}$.
  With this structure, we can interpret a parameterized algebraic theory~$\MCT$ in any $\Set^\bN$-enriched category~$\MCC$ with products, powers, and copowers. A $\MCT$-model in~$\MCC$ comprises an object $X\in \MCC$ together with, for each $\op : (p\vbar m_1\sdots m_k)$, a morphism $\yoneda p\cdot \big( \power{\yoneda {m_1}} X\times \dots\times \power {\yoneda {m_k}} X\big)\to X$, making a diagram commute for each equation in $\MCT$. (Here, we write $\yoneda m\defeq \bN(m,-)$, and $(A\cdot X)$ and $[A,X]$ for the copower and power.)
  The elementary notion of model (Def.~\ref{def:model}) is recovered because, for the symmetric monoidal closed structure on $\Set^\bN$ itself,
  $(\power{\yoneda m} X)(n)=X(n+m)$. This also connects with~(\ref{eqn:later-earlier}), since $(\later X)=\yoneda 1\otimes X$ and $(\earlier X)=\power{\yoneda 1}X$.
\end{remark}

\subsection{Free Models and Monads}\label{sec:free-models}
Strong monads are of fundamental importance to computational effects~\cite{Moggi91}. Algebraic theories give rise to strong monads via free models.

In slightly more detail, there is an evident notion of homomorphism applicable to $\Sigma$-structures and $\MCT$-models, and thus we can sensibly discuss $\Sigma$-structures and $\MCT$-models that are \emph{free} over some collection $X \in \Set^\bN$ of generators.

\begin{definition}\label{def:freeness}
  Consider a $\MCT$-model $\MCY$ with carrier $Y\in\Set^\bN$ and a morphism $\eta_X:X\rightarrow Y$ in $\Set^\bN$.
  Then $\MCY$ is \emph{free on $X$}, if for any other model $\MCZ$ and any morphism $f: X\rightarrow Z$ in $\Set^\bN$, there exists a unique homomorphism of models $\hat f : \MCY \rightarrow \MCZ$, that extends $f$, meaning $\hat f \circ \eta_X = f$ in $\Set^\bN$.
\end{definition}

Informally, for a theory $\MCT$ we define $F_\MCT X \in \Set^\bN$ by taking $F_\MCT X(n)$ to be the set of $=_\MCT$-equivalence classes of terms with parameter context $a_1,\ldots,a_n$ whose $m_i$-ary computation variables come from $X(m_i)$.
More formally, we let
\begin{displaymath}
  F_\MCT X(n) = \{ \tuple{[x_1 : m_1, \ldots, x_k : m_k \vbar a_1, \ldots a_n \vdash t]_{=_\MCT},c_1,\ldots,c_k} \mid c_i \in X(m_i) \} / \sim
\end{displaymath}
where the equivalence relation $\sim$ allows us to $\alpha$-rename context variables in the term judgements and apply permutation, contraction or weakening to the computation context paired with the corresponding transformation of the tuple $c_1,\ldots,c_k$.
It is straightforward to make $F_\MCT X$ into a $\Sigma$-structure.

\begin{proposition}\label{prop:free-models}
~
\begin{enumerate}
\item
  $F_\MCT X$ is a $\MCT$-model, and moreover a free $\MCT$-model over $X$.
\item $F_\MCT$ extends to a monad on $\Set^\bN$, strong for the Day tensor.
\item The derivable equality $(=_\MCT)$ in a parameterized algebraic theory~$\MCT$ is complete:
  if an equation is valid in every $\MCT$-model, then it is derivable in $=_\MCT$.
\end{enumerate}
\end{proposition}
A monad $T$ on $\Set^\bN$ strong for the Day tensor is a monad in the usual sense equipped with a strength $X\otimes TY \rightarrow T(X\otimes Y)$, where $\otimes$ is the Day tensor defined in Rem.~\ref{rem:models-via-enriched-categories}.
\begin{proof}
  The first part of~(1) is straightforward and standard: the interpretation of operation symbols is read off from \eqref{eq:op-former}.
  For the second part, we use the variable introduction rule \eqref{eq:var-former} to define $\eta_X : X \to F_\MCT X$ by sending $c \in X(n)$ to the class of $\tuple{[x : n \vbar a_1, \ldots, a_n \vdash x]_{=_\MCT},c}$.
  Then, for any $\MCT$-model $\MCY$ and map $\phi : X \to Y$, we extend (uniquely) to a homomorphism $F_\MCT X \to \MCY$ by sending an element of $F_\MCT X(n)$ represented by $\tuple{[x_1 : m_1, \ldots, x_k : m_k \vbar a_1, \ldots a_n \vdash t]_{=_\MCT},c_1,\ldots,c_k}$ to  $\den{\vec x \vbar\vec a \vdash t}_{\MCY,0} (\phi_{m_1}(c_1),\ldots,\phi_{m_k}(c_k)) \in Y(n)$.
  Claim (3) follows since all terms are interpreted in free models ``as themselves'', and two terms are derivably equal iff their denotations in the free model are equal.

  For~(2), the $\eta_X$ just defined becomes the monadic unit, and bind is defined in terms of substitution as usual (rule~(\ref{eq:general-subst-rule})).
  By the general properties of the Day tensor product, the strength is determined by a natural transformation
  \begin{displaymath}
    X(p) \otimes F_\MCT Y(n) \to F_\MCT(X \otimes Y)(p + n).
  \end{displaymath}
  This map sends $c \in X(p)$ paired with an element represented by $\tuple{[y_1 : m_1, \ldots, y_k : m_k \vbar b_1, \ldots b_n \vdash t]_{=_\MCT},d_1,\ldots,d_k}$ where $d_i \in Y(m_i)$ to the class of
  \begin{displaymath}
    \tuple{[y_1 : p+m_1, \ldots, y_k : p+m_k \vbar a_1, \ldots, a_p, b_1, \ldots b_n \vdash t']_{=_\MCT},\tuple{\tuple{p,m_1},\tuple{c,d_1}},\ldots,\tuple{\tuple{p,m_k},\tuple{c,d_k}}}
  \end{displaymath}
  where $t' = t[\ldots,(a_1,\ldots,a_p, c_1,\ldots,c_{m_i} \vdash y_i(a_1,\ldots,a_p,c_1,\ldots,c_{m_i}))/y_i, \ldots]$ and the $\tuple{\tuple{p,m_i},\tuple{c,d_i}}$ are elements of $(X \otimes Y)(p + m_i) = \sum_{h_1+h_2=p+m_i} X(h_1)\times Y(h_2)$ that lie within the $X(p) \times Y(m_i)$ summand.
  It is straightforward to check the required laws.
\end{proof}

\todo[inline]{SM: It would be nice to remark that these monads are not strong for the cartesian product - pick a concrete example to show it with maybe.}

In~\Cref{sec:free-models-scoped} we will consider explicit syntax-free characterizations of the free models for particular scoped theories.
For completeness, we mention that $F_\MCT$ is part of an equivalence between such sifted-colimit-preserving strong monads and parameterized theories, e.g.~\cite[\S5]{DBLP:conf/popl/Staton15}.
We omit details as this will not be used in the rest of this paper.

\begin{proposition}
  The functor $F_\MCT$ preserves sifted colimits.
  Moreover, there is an equivalence between parameterized algebraic theories
  and monads on $\Set^\bN$ strong for the Day tensor and whose functor part preserves sifted colimits. 
\end{proposition}

\subsubsection*{Relating parameterized theories and algebraic theories}
In~\Cref{prop:pat-monad-on-set,,thm:noeqn} we relate the free model construction $F_\MCT$ to monads on $\Set$ that support algebraic and scoped operations, as discussed in~\Cref{sec:scop-effects-handl}.
We also discuss an encoding of algebraic theories as parameterized theories~(\Cref{prop:restricts-to-monad-for-alg-theory}), which will be extended to include scoped operations in~\Cref{sec:gener-param-theory}.

Recall from~\Cref{sec:scop-effects-handl} that the functors $\upp{} : \Set \to \Set^\bN$ and $\downn{} : \Set^\bN \to \Set$ form an adjunction.
The monad $F_\MCT$ induces a monad on $\Set$:
\begin{equation}
  \label{eq:18}
  F'_\MCT \coloneqq \downn{\ } \circ F_\MCT \circ \upp{}
\end{equation}

\begin{proposition}\label{prop:pat-monad-on-set}
  Let $\MCT = \tuple{\Sigma,E}$ be a parameterized algebraic theory arising from signatures for algebraic and scoped operations (\Cref{rem:scoped-signature-to-parameterized-signature}). 
  \begin{enumerate}
  \item For each algebraic operation $(o : k)$, the monad $F'_\MCT$ supports an algebraic operation of the same arity, in the sense of  \Cref{def:op-on-monad}.
  \item For every scoped operation $(s : k)$, the monad $F'_\MCT$ supports a scoped operation of the same arity, in the sense of~\Cref{def:scoped-op-on-monad}.
  \end{enumerate}
\end{proposition}
\begin{proof}
  We have a natural transformation $(F_\MCT X)^k \to F_\MCT X$ given by substituting terms into the operation $o$ and this is ``algebraic'' in the sense of commuting with postcomposition by the monadic multiplication, so (1) follows easily.
  In general, an operation $s : (0 \vbar 1,\sdots,1)$ appears as a natural transformation
  $(\earlier F_\MCT X)^k \to F_\MCT X$.
  However, it is easy to show that $\earlier$ preserves $\MCT$-models, and in the setting of (2) we have can use $\close : (1 \vbar 0)$ to get a map $X \xrightarrow{\eta_X} F_\MCT X \xrightarrow{\den\close} \earlier F_\MCT X$.
  By the universal property of free models, this induces another map $F_\MCT X \to \earlier F_\MCT X$.
  This map adds a $\close$ at the leaves of each syntax tree, as opposed to at the root.
  Postcomposing $F_\MCT X \to \earlier F_\MCT X$ with the scoped operation $(\earlier F_\MCT X)^k \to F_\MCT X$, we have a natural transformation $(F_\MCT X)^k \to F_\MCT X$ which restricts to the desired scoped operation on $F'_\MCT$.
\end{proof}

\Cref{prop:pat-monad-on-set} (1) is still true when applied to operations with arity $(0 \vbar 0, \sdots, 0)$ part of an \emph{arbitrary} parameterized signature.

Key to~\Cref{persp:effect-as-alg-theory} is that an algebraic theory determines a monad on $\Set$.
An ordinary algebraic theory $\MCT$ can be considered as a parameterized algebraic theory in at least two ways.
Let $\MCT_1$ be the ``minimal translation'': the parameterized theory whose signature consists only of operations of arity $(0\vbar 0, \sdots, 0)$ plus all equations; and let $\MCT_2$ be $\MCT_1$ together with an operation $\close : (1 \vbar 0)$.
Thus $\MCT_2$ arises when translating into a parameterized theory via a scoped signature as in~\Cref{rem:scoped-signature-to-parameterized-signature}, albeit with no scope-opening operations.

\begin{proposition}\label{prop:restricts-to-monad-for-alg-theory}
  Let $\MCT = \tuple{\Sigma,E}$ be an algebraic theory and inducing a monad $T$ on $\Set$.
  Then $T$ is isomorphic to the $\Set$ restriction of the monads induced by the parameterized theories $\MCT_1$ and $\MCT_2$ (defined in~\cref{eq:18}), respectively:
  \begin{displaymath}
    T \cong F'_{\MCT_1} \cong F'_{\MCT_2}.
  \end{displaymath}
\end{proposition}
\begin{proof}
  We will prove something stronger.
  It is easy to see that for any $X \in \Set^\bN$, $(F_{\MCT_1} X)(n) \cong T(X(n))$, since a $\MCT_1$-algebra structure on $Y \in \Set^\bN$ is equivalent to a $\MCT$-algebra structure on each $Y(n)$ with no compatibility requirements.
  Then by definition of the functors $\upp{}$ and $\downn{\ }$ we obtain that $F'_{\MCT_1}(A)\cong T(A)$ for any set $A$.

  For $\MCT_2$, we can at least characterize $F_{\MCT_2}(\upp{A})$, for a set $A$.
  It is given by $\bar T A \coloneqq \lambda n.\, T^{n+1} A$, with the obvious interpretation of the algebraic operations in $\Sigma$, and $\close_{\bar T A,n} : T^{n+1} A \to T^{n+2} A$ given by the monadic unit $\eta_{T^{n+1} A}$.
  To see that $\bar T A$ is the free model of $\MCT_2$, given any $\MCT_2$-model $Y$ with a map $\upp{A} \to Y$, we define a map $\bar T A \to Y$ as follows.
  Since each $Y(n)$ is, in particular, a $T$-algebra, we get a unique $T$-homomorphism $TA \to Y(0)$.
  Now we can inductively construct $T$-homomorphisms $T^{n+2} A \to Y(n+1)$: composing $T^{n+1} A \to Y(n)$ constructed so far with $\close_{Y,n} : Y(n) \to Y(n+1)$, and extending uniquely to a $T$-homomorphism $T^{n+2} A \to Y(n+1)$.
  The resulting map $\bar T A \to Y$ a $\MCT_2$-homomorphism since it commutes with the interpretations of $\close$ by construction.
\end{proof}

This establishes the connection with monads for ordinary algebraic theories.
Finally, we will connect $F_\MCT$ with the scoped monad of Prop.~\ref{prop:pswj-monads-isomorphic} that was first given by \citet{DBLP:conf/lics/PirogSWJ18}.

\begin{theorem}\label{thm:noeqn}
  Consider signatures for algebraic $\Sigma$ and scoped $\Sigma'$ effects with no equations,
  inducing a parameterized algebraic theory $\MCT$  (via Rem.~\ref{rem:scoped-signature-to-parameterized-signature}).
  We have an isomorphism of monads \[F_\MCT\cong \paren{ \bar{\Sigma} + \paren{ \bar{\Sigma'} \circ \earlier } + \later }^*.\]
\end{theorem}
\begin{proof}
  To see this, we use the description of $F_\MCT X(n)$ as a set of equivalence classes of $\MCT$-terms with computation variables coming from $X$.
  We show that $F_\MCT X(n)$ is isomorphic to the set $\paren{ \bar{\Sigma} + \paren{ \bar{\Sigma'} \circ \earlier } + \later }^*(X)(n)$.
  Consider the outermost operation of a term in $F_\MCT X(n)$: it is either algebraic, scoped or $\close$. 
  Each component of the sum on the right-hand-side, $\bar{\Sigma}$, $\paren{ \bar{\Sigma'} \circ \earlier }$ and $\later$  corresponds to one of these three possibilities.
  Scoped operations bind a parameter hence the need for the $\earlier$ functor which increases the index $n$ by $1$.
  The $\close$ operation consumes a parameter which corresponds to the fact that $\later$ decreases the index $n$ by $1$.
  Note that according to~\Cref{def:sigma-struct}, for an arbitrary carrier $Y$, the interpretation of $\close$ has type $Y\rightarrow \earlier Y$, which is equivalently a morphism  $\later Y \rightarrow Y$.
  Both $\earlier$/$\later$ are characterized in Rem.~\ref{rem:models-via-enriched-categories} in terms of the Day tensor of $\Set^\bN$.
\end{proof}

  

\subsection{Free Models for Scoped Effects}\label{sec:free-models-scoped}

We now turn to some concrete models from \citet{DBLP:conf/lics/PirogSWJ18}.
To characterize them as certain free models of parameterized algebraic theories, we need the following notion.
\begin{definition}
  $X \in \Set^\bN$ is \emph{truncated} if $X(n+1) = \emptyset$ for all $n \in \bN$.
\end{definition}
Equivalently, $X$ is truncated if $X = \upp (X(0))$.
The free model on a truncated $X$ corresponds to the case where computation variables can only denote programs with no open scopes.
This is the case in the development of \citet{DBLP:conf/lics/PirogSWJ18}, where if the programmer opens a scope, a matching closing of the scope is implicitly part of the program.

\subsubsection{Nondeterminism with $\once$}\label{sec:nond-with-once-model}
Recall the parameterized theory for nondeterminism with $\once$ (signature in Ex.~\ref{ex:nondet-once-sig} and equations in Fig.~\ref{fig:nondet-once}).
It follows from Prop.~\ref{prop:free-models} that this theory has a free model on each $X$ in $\Set^\bN$, with carrier denoted by $\tonce(X) \in \Set^\bN$.
For $X$ truncated, the free model on $X$ has an elegant description:
\begin{displaymath}
  \tonce(X)(n) = \listf^{n+1}(X(0)).
\end{displaymath}
Here $\listf(A)$ is the set of lists over $A$. (So $\listf^{n+1}(X(0))$ can be thought of as the set of balanced trees of depth exactly~$n+1$ but with arbitrary degree, and leaves labelled by $X(0)$, where we distinguish between nodes with zero degree and leaf nodes.)

In this case the interpretation of $\once$ chooses the first element of a list and closing a scope wraps its continuation as a singleton list.
Choice is interpreted as list concatenation ($\concat$), and failure as the empty list $[]$:
\begin{align*}
  &\once_n: \tonce(X)(n+1)\rightarrow \tonce(X)(n)  &&\once_n([]) = [],\ \once_n([x,\ldots]) = x \\
  &\close_n :\tonce(X)(n) \rightarrow\tonce(X)(n+1)\ &&\close_n(x)=[x] \\
  &\orop_n :\tonce(X)(n) \times \tonce(X)(n) \rightarrow \tonce(X)(n) &&\orop_n(x_1,x_2) = x_1 \concat x_2 \\
  &\fail_n : 1\rightarrow \tonce(X)(n) &&\fail_n()=[]
\end{align*}

In fact the model $\tonce(X)$ we just described is the same as the model for nondeterminism of \citet[Example 4.2]{DBLP:conf/lics/PirogSWJ18}:
\begin{theorem}\label{thm:nondet-free-model}
  For a truncated $X\in\Set^\bN$, $\tonce(X)$ is the free model of the parameterized theory of nondeterminism with $\once$~(\Cref{exa:nondet-once-eq} and \Cref{fig:nondet-once}).
  Moreover, the model for nondeterminism with $\once$ from~\cite[Example~4.2]{DBLP:conf/lics/PirogSWJ18}, starting from a set $A$, is the free model on $\upp A\in\Set^\bN$.
\end{theorem}
\begin{proof}[Proof notes]
  It is easy to show that the operations on $\tonce(X)$ defined in this section satisfy the equations in~\Cref{fig:nondet-once}, and therefore $\tonce(X)$ is a model.
  Denote by $\freeonce(X)$ the free model in the sense of~\Cref{prop:free-models}.
  We apply the definition of freeness to the map $X\rightarrow \tonce(X)$ that sends each $c\in X(0)$ to the singleton list $[c]$, to obtain a homomorphism of models $\rho:\freeonce(X)\rightarrow \tonce(X)$.
  To show $\tonce(X)$ is free, we show that $\rho$ has an inverse $\sigma:\tonce(X)\rightarrow \freeonce(X)$.

  Recall that $\freeonce(X)(n)$ contains equivalence classes of terms in context, with computation variables taken from $X$.
  We use representatives of these equivalence classes where the context $\Gamma_X$ contains a variable $(x:0)$ for each element of $X(0)$.
  We define $\sigma$ by induction:
  \begin{align*}
    &\sigma_n([]) = \Gamma_X \vbar a_1,\sdots,a_n \vdash \fail \\
    &\sigma_0(\mathsf{x::xs}) = \Gamma_X \vbar - \vdash \orop(\mathsf{x},\sigma_0(\mathsf{xs}))\\
    &\sigma_{n+1}(\mathsf{x::xs}) = \Gamma_X \vbar a_1,\sdots,a_{n+1} \vdash \orop\bigl(\close(a_{n+1},\sigma_n(x)),\sigma_{n+1}(\mathsf{xs}) \bigr)
  \end{align*}
  The terms in the image of $\sigma$ are essentially normal forms for the parameterized theory, when computation variables only have arity $0$.

  We show that $\sigma$ is a homomorphism of models, meaning it commutes with the operations in the theory, using the equations in~\Cref{fig:nondet-once}.
  The case for $\once$ also requires an induction on lists.

  To show $\rho\circ\sigma = \mathsf{id}_{\tonce(X)}$, we use induction on $n\in\bN$, followed by induction on lists, and the fact that $\rho$ is a homomorphism.
  To show $\sigma\circ\rho = \mathsf{id}_{\freeonce(X)}$, we do induction on the term formation rules, and use the fact that $\sigma$ is a homomorphism.
\end{proof}

\subsubsection{Exceptions}
Recall the parameterized theory of throwing and catching exceptions in Ex.~\ref{exa:catch-signature} and (\ref{eq:8}--\ref{eq:9}).
For truncated $X\in\Set^\bN$, the free model of the theory of exceptions has carrier:
\begin{displaymath}
  \tcatch(X)(n) = X(0) + \{e_0,\dots,e_n\}
\end{displaymath}%
where $e_{n-i}$ corresponds to the term (in normal form) that closes $i$ scopes then throws.

To define the operations $\catch_n$ and $\close_n$ we pattern match on the elements of $\tcatch(X)(n+1)$ using the isomorphism $\tcatch(X)(n+1)\cong \tcatch(X)(n)+\{e_{n+1}\}$.
Below, $x$ is an element of $\tcatch(X)(n)$, standing for a computation in normal form: 
\[
  \begin{array}{l}
    \begin{aligned}
      &\catch_n:\tcatch(X)(n+1)\times \tcatch(X)(n+1) \rightarrow \tcatch(X)(n) \\
      &\catch_n(x,-)= x, \
      \catch_n(e_{n+1},\,x) = x,\
      \catch_n(e_{n+1},\,e_{n+1})=e_n 
    \end{aligned}\\ \\%
    \begin{aligned}
      &\close_n:\tcatch(X)(n)\rightarrow \tcatch(X)(n+1) && \close_n(x)=x \\
      &\throw_n:1\rightarrow \tcatch(X)(n) && \throw_n()=e_n
    \end{aligned}
  \end{array}
\]
The cases in the definition of $\catch_n$ correspond to  equations (\ref{eq:9}), (\ref{eq:8}), (\ref{eq:10}) respectively.
In the third case, an exception  inside $n+1$ scopes in the second argument of $\catch$ becomes an exception inside $n$ scopes.

\begin{theorem}\label{thm:catch-free-model}
  For a truncated $X\in\Set^\bN$, $\tcatch(X)$ is the free model for the parameterized theory of exceptions (\ref{eq:8}--\ref{eq:9}).
  The model for exception catching from~\cite[Example 4.5]{DBLP:conf/lics/PirogSWJ18}, which starts from a set $A$, is  the free model on $\upp A \in\Set^\bN$.
\end{theorem}
\begin{proof}[Proof notes]
  We follow the same outline as for the proof of~\Cref{thm:nondet-free-model}.
  Consider the map $X\rightarrow \tcatch(X)$ that maps each element of $X(0)$ to itself.
  This has a unique extension to a map out of the free model on $X$, $\rho:\freecatch(X)\rightarrow\tcatch(X)$.
  We define $\sigma:\tcatch(X)\rightarrow\freecatch(X)$ to be the candidate inverse for $\rho$.
  Note that it is enough to show that $\sigma$ is an inverse in $\Set^\bN$.
  Using the intuition of normal forms and the isomorphism $\tcatch(X)(n) \cong X(0) + \{e_0,\dots,e_n\}$:
  \begin{align*}
    &\sigma_n(x) = \Gamma_X \vbar a_1,\sdots,a_n \vdash \close(a_n,\sdots \close(a_1,x)\sdots) \\
    &\sigma_n(e_{n-i}) =\Gamma_X \vbar a_1,\sdots,a_n \vdash \close(a_{n+1},\sdots \close(a_{n-i+1},\throw) \sdots), \quad \text{for each } 0\leq i\leq n
  \end{align*}

  If we consider an element $z\in\tcatch(X)(n)$, and the isomorphism $\tcatch(X)(n+1)\cong \tcatch(X)(n)+\{e_{n+1}\}$, the following equation holds:
  \begin{equation}\label{eq:17}
    \sigma_{n+1}(z) = \Gamma_X \vbar a_1,\sdots,a_{n+1} \vdash \close(a_{n+1},\sigma_n(z))
  \end{equation}

  To show $\rho\circ\sigma = \mathsf{id}_{\tcatch(X)}$ we use that $\rho$ is a homomorphism.
  For $\sigma\circ\rho = \mathsf{id}_{\freecatch(X)}$, we proceed by induction on the term formation rules.
  We use~\cref{eq:17}, and in the $\catch$ case we also use the~\crefrange{eq:8}{eq:9} from the parameterized theory.
\end{proof}

\subsubsection{State with local values}
Recall the parameterized theory of mutable state with local values in Ex.~\ref{ex:state-local-sig} and its equations (\ref{eq:5}--\ref{eq:7}).
The free model, in the sense of Prop.~\ref{prop:free-models}, on a truncated $X\in \Set^\bN$ has carrier:
  \begin{align*}
    \tlocal(X)(0) &= \big(\Bit\Rightarrow X(0)\times \Bit\big) &
    \tlocal(X)(n+1) &= \big(\Bit\Rightarrow \tlocal (X)(n))\big)
  \end{align*}
  The operations on this model are
  \begin{align*}
    &\local{i}_n : \tlocal(X)(n+1)\rightarrow \tlocal(X)(n) && \local{i}_n(f)= (f\,i) \\
    &\close_n : \tlocal(X)(n)\rightarrow \tlocal(X)(n+1) && \close_n(f) = \lambda s.\,f \\
    &\putop{i}_n:\tlocal(X)(n) \rightarrow \tlocal(X)(n) && \putop{i}_n(f) = \lambda s.\,f\,i  \\
    &\get_n: \tlocal(X)(n)^2 \rightarrow \tlocal(X)(n) && \get_n (f,g)\ =\  \lambda s.\ \  \raisebox{0pt}[0pt]{$\begin{cases}f\ s & s = 0 \\ g\ s & \text{otherwise}\end{cases}$}
  \end{align*}
  Notice that the continuation of $\local{i}$ uses the new state $i$, whereas $\close$ discards the state $s$ which comes from the scope that is being closed.

  If we only consider equations~(\ref{eq:5}--\ref{eq:6}), omitting~\eqref{eq:7}, the carrier of the free model on a truncated $X\in\Set^\bN$ is:
  \begin{mathpar}
    \tlocal'(X)(0) = \bigl(\Bit\Rightarrow X(0)\times \Bit\bigr)  = \tlocal(X)(0)\and
    \tlocal'(X)(n+1) = \bigl(\Bit\Rightarrow \tlocal' (X)(n) \times \Bit\bigr)
  \end{mathpar}
  In fact, $\tlocal'(X)$ is the model of state with $\local{}$ proposed in~\cite[\S 7.1]{DBLP:conf/lics/PirogSWJ18}:

  \begin{theorem}\label{thm:state-free-model}
    Given a truncated $X\in\Set^\bN$, $\tlocal(X)$ is the free model of the parameterized theory of state with local values,~\crefrange{eq:5}{eq:7}.
    Moreover, $\tlocal'(X)$ is the free model for~\crefrange{eq:5}{eq:6}.
    
    Consider the example of state with local variables from \citet[\S 7.1]{DBLP:conf/lics/PirogSWJ18}, specialized to one memory location storing one bit, reading the return type (there `$a$') as a set $A$.
    The model proposed in~\cite[\S 7.1]{DBLP:conf/lics/PirogSWJ18} is the free model on $\upp A$ for the parameterized theory with equations~(\ref{eq:5}--\ref{eq:6}).
  \end{theorem}

  \begin{proof}[Proof notes]
    We show that $\tlocal(X)$ is the free model following the same proof outline as for~\Cref{thm:nondet-free-model,thm:catch-free-model}.
    It is easy to equip $\tlocal'(X)$ with operations and do a similar proof of freeness.

    Consider the map $X\rightarrow \tlocal(X)$ that sends each $x\in X(0)$ to the function $\lambda s.\,(x,s)$, and let $\rho:\freelocal(X)\rightarrow\tlocal(X)$ be its unique extension out of the free model.
    We define a candidate inverse for $\rho$ using the intuition of normal forms, $\sigma:\tlocal(X)\rightarrow\freelocal(X)$:
    \begin{align*}
      &\sigma_0(f) = \Gamma_X \vbar - \vdash \get\bigl( \putop{\pi_2(f\, 0)}(\pi_1(f\, 0)),\, \putop{\pi_2(f\, 1)}(\pi_1(f\, 1)) \bigr) \\
      &\sigma_{n+1}(f) = \Gamma_X \vbar a_1,\sdots,a_{n+1} \vdash \get\bigl( \close(a_{n+1},\sigma_n(f\,0)),\, \close(a_{n+1},\sigma_n(f\,1)) \bigr)
    \end{align*}
    where we identified the elements of $X(0)$ with the variables in $\Gamma_X$.

    We can show that $\sigma$ is a homomorphism using~\crefrange{eq:5}{eq:7} from the parameterized theory.
    To show $\rho\circ\sigma = \mathsf{id}_{\tlocal(X)}$, we proceed by induction on $n\in\bN$ and use the fact that $\rho$ is a homomorphism.
    For $\sigma\circ\rho = \mathsf{id}_{\freelocal(X)}$, we use induction on terms and that both $\rho$ and $\sigma$ are homomorphisms.
    For the variable case, we use~\cref{eq:5}.
  \end{proof}

  The interpretations in $\tlocal(X)$ and $\tlocal'(X)$, i.e~that of \citet{DBLP:conf/lics/PirogSWJ18}, of programs with no open scopes agree:
  \begin{proposition}\label{prop:state-local-interp-coincide}
  Consider a fixed context of computation variables $\Gamma=(x_1:0,\sdots,x_n:0)$ and a truncated $X\in\Set^\bN$.
  For any term $\Gamma\vbar - \vdash t$, the following two interpretations coincide at index $0$:
  \begin{displaymath}
    \den{t}_{\tlocal(X),0} = \den{t}_{\tlocal'(X),0} : \tlocal(X)(0)^n\rightarrow \tlocal(X)(0),
  \end{displaymath}
  under the identification $\tlocal(X)(0)=\tlocal'(X)(0)$.
\end{proposition}
\begin{proof}
  The $\local{}$ operation binds parameters so we use a logical relation indexed by the length of the parameter context $p\in\bN$.
  The relation has type
  \begin{displaymath}
    \MCR^p \subseteq \bigl(\tlocal X(0)^n \rightarrow \tlocal X(p)\bigr) \times \bigl(\tlocal X(0)^n \rightarrow \tlocal' X(p)\bigr)
  \end{displaymath}
  and is defined as
  \begin{align*}
    \MCR^0 &= \{(f,f) \vbar f: \tlocal X(0)^n \rightarrow \tlocal X(0)\} \\
    \MCR^{p+1} &= \bigl\{(f,g) \vbar 
    \forall s\in\Bit.\,\bigl(\lambda(x_1,\sdots,x_n).\,f(x_1,\sdots,x_n)(s),\  \lambda(x_1,\sdots,x_n).\,\pi_1g(x_1,\sdots,x_n)(s)\bigr) \in \MCR^p \bigr\}
  \end{align*}
  where $(x_1,\sdots,x_n)\in \tlocal X(0)^n$.

  The fundamental property for this logical relation says that for all terms $\Gamma \vbar a_1,\sdots,a_p \vdash t$, the two interpretations are related: $\bigl(\den{t}_{\tlocal(X),0}, \den{t}_{\tlocal'(X),0}\bigr) \in\MCR^p$.
  If $p=0$, this is enough to deduce the interpretations are equal.
  The fundamental property is proved by induction on $t$.
\end{proof}

The restriction of $\Gamma$ in~\Cref{prop:state-local-interp-coincide} to computation variables that do not depend on parameters and of $\Delta$ to be empty are reasonable because in the framework of \citet{DBLP:conf/lics/PirogSWJ18}, only programs with no open scopes are well-formed.
Therefore, only such programs can be substituted in $t$, justifying the restriction of $\den{t}_{\tlocal'(X)}$ to index $0$.

\subsubsection{Nondeterminism with $\cut$}\label{sec:nond-with-cut}
Recall the parameterized theory of explicit nondeterminism with $\cut$ from~\Cref{ex:cut-scope}.
According to~\Cref{prop:free-models}, this theory has a free model for each $X\in\Set^\bN$, denoted by $\tcut(X)$.

Define a functor on sets $\idemf$:
\begin{displaymath}
  \idemf(A) = A \uplus \{a^* \vbar a\in A\}.  
\end{displaymath}
Here in the notation $a^*$, the $*$ is just a flag. It is a semantic counterpart of $\cut$, meaning `discard the yet uninspected choices and continue with $a$',
following~\cite[\S6]{DBLP:journals/jfp/PirogS17}. (It is not the Kleene star of Prop.~\ref{prop:pswj-monads-isomorphic}.)

If $X$ is truncated, the carrier of the free model on $X$ is:
\begin{displaymath}
  \tcut(X)(n) = (\idemf\circ\listf)^{n+1}(X(0)).
\end{displaymath}
(So these can be thought of as balanced binary trees, depth $n+1$, leaves labelled by $X(0)$, and where each non-leaf node is optionally flagged with $*$.)

When writing the interpretation of operations, we use the isomorphism
\begin{mathpar}
  \tcut(X)(0) \cong \idemf(\listf(X(0)))) \and \tcut(X)(n+1) \cong \idemf(\listf(\tcut(X)(n)))
\end{mathpar}
so an element of $\tcut(X)(n)$ is either a list $\mathsf{xs}$ or a starred list $\mathsf{xs}^*$.

The interpretation of $\cut$ and $\orop$ in the free model is the following:
\begin{displaymath}
  \begin{array}{l}
    \begin{aligned}
      &\cut_n : \tcut(X)(n)\rightarrow \tcut(X)(n) && \cut_n(\mathsf{xs}) = \mathsf{xs}^*,\quad \cut_n(\mathsf{xs}^*) = \mathsf{xs}^* \\
    \end{aligned}\\ \\%
    \begin{aligned}   
  &\orop_n : \tcut(X)(n)\times \tcut(X)(n) \rightarrow \tcut(X)(n) \\ 
  &\orop_n(\mathsf{xs},\mathsf{ys}) = \mathsf{xs}\concat\mathsf{ys},\quad
    \orop_n(\mathsf{xs},\mathsf{ys}^*) = (\mathsf{xs}\concat\mathsf{ys})^*, \quad
    \orop_n(\mathsf{xs}^*,-) = \mathsf{xs}^*
    \end{aligned}
  \end{array}
\end{displaymath}
The $\cut$ operation marks a list with a star; the second clause in the definition of $\cut$ corresponds to the idempotence equation~(\ref{eq:13}).
Similarly to~\Cref{sec:nond-with-once-model}, $\orop$ performs list concatenation, but now stars need to be taken into account as well.
The second clause in definition of $\orop$ corresponds to~\cref{eq:14}, and the third clause to~\cref{eq:15}.

The operation of opening a scope is interpreted as:
\begin{displaymath}
  \begin{aligned}
    &\scope_n : \tcut(X)(n+1) \rightarrow \tcut(X)(n) \\
    &\scope_n(\mathsf{xs^*}) = \scope_n(\mathsf{xs}), \quad
    \scope_n([]) = [], \quad
    \scope_n(\mathsf{x::xs}) = \orop_n(\mathsf{x},\scope_n(\mathsf{xs}))
  \end{aligned}
\end{displaymath}
The first clause corresponds to erasing cuts as in~\cref{eq:11}.
The second clause corresponds to~\cref{eq:16}.
The third clause uses the interpretation of $\orop$ to concatenate lists while taking into account cuts; it corresponds to~\cref{eq:12}.

The $\close$ and $\fail$ operations have the same interpretation as in the nondeterminism with $\once$ example from~\Cref{sec:nond-with-once-model}:
\begin{align*}
  &\close_n : \tcut(X)(n) \rightarrow \tcut(X)(n+1) && \close_n(x) = [x] \\
  &\fail_n : 1\rightarrow \tcut(X)(n) &&\fail_n() = []
\end{align*}

\todo[inline]{C: I'm not sure how to formulate a theorem that says this model is essentially the same as the one in \citet{DBLP:journals/jfp/PirogS17} for $n=0$.}

\begin{theorem}\label{thm:cut-free-model}
  If $X\in\Set^\bN$ is truncated, then $\tcut(X)$ is the free model of the parameterized theory of nondeterminism with $\cut$, defined in~\Cref{ex:cut-scope}.
\end{theorem}
\begin{proof}
  We explained in this subsection why $\tcut(X)$ together with its interpretation of operations respects the equations of the parameterized theory.
  Therefore, $\tcut(X)$ is a model of the theory, in the sense of~\Cref{def:model}.

  Denote by $\freecut(X)$ the free model on $X$ for nondeterminism with $\cut$, in the sense of~\Cref{prop:free-models}.
  Consider the map $X\rightarrow \tcut(X)$ that sends each element of $X(0)$ to the singleton list containing that element.
  By freeness~(\Cref{def:freeness}), this map has a unique extension to a homomorphism of models $\rho:\freecut(X)\rightarrow \tcut(X)$.
  To show $\tcut(X)$ is free, it is enough to show that $\rho$ has an inverse. 

  Recall that $\tcut(X)(n)$ is isomorphic to $(\idemf\circ\listf)^{n+1}(X(0))$, and $\freecut(X)(n)$ contains equivalence classes of terms in context, where the computation variables come from $X$.
  We define a candidate inverse $\sigma:\tcut(X)\rightarrow\freecut(X)$ by induction on $n\in\bN$ and on lists.
  For simplicity, we use representatives of equivalence classes from $\freecut(X)(n)$ where the context $\Gamma_X$ contains one variable $x_i:0$ for each $c_i\in X(0)$, so we identify $x_i$ and $c_i$ in the definition below:
  \begin{displaymath}
    \arraycolsep=0pt
    \begin{array}{rl}
      \sigma_0(\mathsf{xs}^*) &{}= \Gamma_X\vbar - \vdash \cut(\sigma_0(\mathsf{xs}))
    \quad \text{where } \mathsf{xs} \in\listf(X(0))  
    \\
      \sigma_0([]) &{}= \Gamma_X\vbar - \vdash \fail \\
      \sigma_0(\mathsf{x::xs}) &{}= \Gamma_X\vbar -\vdash \orop(\mathsf{x},\sigma_0(\mathsf{xs}))
    \\[6pt]
      \sigma_{n+1}(\mathsf{xs}^*) &{}= \Gamma_X \vbar a_1,\sdots,a_{n+1} \vdash \cut(\sigma_{n+1}(\mathsf{xs})) \\
      \sigma_{n+1}([]) &{}= \Gamma_X \vbar a_1,\sdots,a_{n+1} \vdash \fail \\
      \sigma_{n+1}(\mathsf{x::xs}) &{}= \Gamma_X \vbar a_1,\sdots,a_{n+1} \vdash\orop\bigl(\close(a_{n+1},\sigma_n(\mathsf{x})),\,\sigma_{n+1}(\mathsf{xs}) \bigr) 
   \end{array}
\end{displaymath}
The terms in the image of $\sigma$ can be seen as the normal forms for the parameterized theory, defined recursively.

To show that $\sigma$ is a section, $\rho\circ\sigma = \mathsf{id}_{\tcut(X)}$, we use induction on $n\in\bN$.
For both the base case and the induction step, we first consider the case where the input is a list $\mathsf{xs}$ and use induction on lists.
We use this intermediate result to prove the case of a starred list $\mathsf{xs}^*$.
Throughout, we use the fact that $\rho$ is a homomorphism.

To show $\sigma\circ\rho = \mathsf{id}_{\freecut}$, we fist show $\sigma$ is a homomorphism of models, i.e~that it commutes with all the operations.
The $\orop$ and $\scope$ cases require induction on lists.
Most importantly, we make use of equations from the parameterized theory of cut~(\Cref{ex:cut-scope}).
We then show by induction on the term formation rules that:
\begin{displaymath}
  \sigma_n\bigl(\rho_n(\Gamma_X\vbar a_1,\sdots,a_n\vdash t)\bigr) =_{\mathsf{cs}} \bigl( \Gamma_X\vbar a_1,\sdots,a_n\vdash t \bigr)
\end{displaymath}
where the equality $=_{\mathsf{cs}}$ is in the theory of cut.
We use the equations in the theory and the fact that $\sigma$ is a homomorphism.
The $\scope$ case also requires an induction on lists.
\end{proof}
\todo[inline]{C: maybe this proof has too much detail?}

\subsection{Generating a Parameterized Theory from an Arbitrary Scoped Operation}\label{sec:gener-param-theory}





In~\Cref{rem:scoped-signature-to-parameterized-signature} we showed how to translate signatures of arbitrary algebraic and scoped~(\Cref{def:alg:sigs,sec:scop-effects-handl}) operations into parameterized signatures.
We then showed in~\Cref{sec:pat:ex} how to encode the behaviour of several \emph{examples} of scoped operations using equations that form parameterized theories.
In~\Cref{sec:free-models-scoped}, we constructed free models for these examples of scoped operations.
In this section, we show how to obtain a parameterized theory from an arbitrary
scoped operation, and discuss the free model of such a theory.

Recall from~\Cref{persp:scoped-effects} that a scoped operation is one that does not commute with substitution, and more precisely from~\Cref{def:scoped-op-on-monad} that a scoped operation on a monad $T$ on $\Set$ is a natural transformation $\sigma_A : (TA)^k \to TA$.

  

Starting from an algebraic theory $\MCT = \tuple{\Sigma, E}$, with monad $T$ supporting a scoped operation $\sigma$ with arity $k$, we can construct a parameterized theory $\thparam$ in the sense of~\Cref{def:presentation}.
The signature of $\thparam$ is constructed using~\Cref{rem:scoped-signature-to-parameterized-signature} applied to the algebraic signature $\Sigma$, first extended to a scoped signature with one scoped operation symbol $\mathsf{sc}$ (standing for $\sigma$) with arity $(0\vbar 1\sdots 1)$, where the list $1\sdots 1$ has length k.

The equations of the parameterized theory $\thparam$ come in two kinds.
First, each equation in $E$ gives an equation in $\thparam$ in the evident way.
Second, for each $n \in \bN$ and each $k$-tuple of $\Sigma$-terms $(x_1,\sdots,x_n \vdash t_1,\sdots,t_k)$, the $\MCT$-equality class of those terms gives an argument for $\sigma_n : T(n)^k \to T(n)$, and $\sigma_n([t_1],\ldots,[t_k]) = [t']$ for some $x_1,\sdots,x_n \vdash t'$.
For each such argument $\thparam$ gets an equation
\begin{equation}
  x_1:0,\sdots, x_n:0 \vbar - \vdash \mathsf{sc}(a.t_1[\close(a,x_i)/x_i]_{1\leq i\leq n}, \dots, a.t_k[\close(a,x_i)/x_i]_{1\leq i\leq n}) = t'.\label{eq:scoped-on-ground}
\end{equation}



Let $F_{\thparam} X$ be the free model over $X$ of the parameterized theory $\thparam$ from~\Cref{prop:free-models}.
Both $X$ and $F_{\thparam} X$ are objects in $\Set^\bN$.
From the free model of $\thparam$ we can recover the free model of the algebraic theory $\MCT$ that we started with.

\begin{theorem}\label{thm:scoped-theory-from-pat}
  Let $\MCT = \tuple{\Sigma,E}$ be an algebraic theory inducing a monad $T$ on $\Set$, supporting a scoped operation $\sigma$ of arity $k$.
  With $\thparam$ constructed as above, we have the following isomorphism:
  \begin{displaymath}
    T \cong F'_{\thparam} = \downn{\ } \circ F_{\thparam} \circ \upp{}.
  \end{displaymath}
  Moreover, the scoped operation on $F'_{\thparam}$ corresponding to $\mathsf{sc}$, as constructed in~\Cref{prop:pat-monad-on-set}, agrees with the operation $\sigma$.

  


\end{theorem}

\begin{proof}
  In fact, we will show that the free model $F_{\thparam}$ on a truncated object $\upp{A}$, for a set $A$, is $\bar T A \coloneqq \lambda m.\, T^{m+1} A$, from the proof of \Cref{prop:restricts-to-monad-for-alg-theory}.
  We first note that $\bar T A$ supports an interpretation of $\mathsf{sc}$ by
  the follow morphism:
  \begin{displaymath}
    \bigl((\bar T A)(m+1)\bigr)^k = \bigl(T^{m+2}(A))\bigr)^k \xrightarrow{\sigma_{T^{m+1}(A)}} T^{m+2}(A) \xrightarrow{\mu_{T^{m+1}(A)}} T^{m+1}(A) = (\bar T A)(m).
  \end{displaymath}
  
  To check that $\bar T A$ is a model of $\thparam$, we need to see that all instances of \eqref{eq:scoped-on-ground} are validated.
  Note that, for terms $x_1:0,\sdots,x_n:0 \vbar - \vdash t_i$,
  \begin{displaymath}
    \den{t_i[\close(a,x_i)/x_i]_{1 \leq i \leq n}}_m
    =
    (T^{m+1}A)^n \xrightarrow{\den{t_i}_m} T^{m+1}(A) \xrightarrow{T\eta_{T^m A}} T^{m+2}(A)
  \end{displaymath}
  and $\mu \circ \sigma \circ T\eta = \mu \circ T \eta \circ \sigma = \sigma$ by naturality of $\sigma$ and the monad laws.

  Since $\bar T A$ is free on $\upp{A}$ for the theory $\MCT_2$ of~\Cref{prop:restricts-to-monad-for-alg-theory}, which omits $\mathsf{sc}$ and the equations~\eqref{eq:scoped-on-ground}, we only need to show that a $\MCT_2$-homomorphism $\bar T A \to Y$ into a $\thparam$-model is always a \mbox{$\thparam$-homomorphism}, i.e.\ that it commutes with $\mathsf{sc}$.
  But every element of $\bar T A(m+1) = T(\bar T A(m))$ is in the image of $\den t_{m+1} \circ \den \close_m$ for some $\Sigma$-term $x_1,\sdots,x_n\vdash t$, thus this amounts to the hypothesis that $Y$ validates all instances of \cref{eq:scoped-on-ground}.
\end{proof}





\begin{remark}
  Whereas in~\Cref{sec:pat:ex} we constructed scoped theories out of particular algebraic theories whose monads support a scoped operation, the recipe for $\thparam$ constructs a scoped theory in the general case.
  However, the presentation of $\thparam$ typically includes infinitely many equations whereas the theories in \Cref{sec:pat:ex} are given finite presentations.
  For example, in the case of $\once$~(\Cref{fig:nondet-once}),~\cref{eq:4b,eq:4} are not part of $\thparam$ but infinitely many substitution instances of them are.
  The free models of $\thparam$ and that in~\Cref{sec:nond-with-once-model} agree on truncated $X$, (since the free model of the algebraic theory of explicit nondeterminism is $\listf$), which shows that the theory of $\once$ admits derivations of all equations of the corresponding $\thparam$.
  However, \cref{eq:4b,eq:4} are not derivable in $\thparam$.
  We argue that the full behaviour of $\once$ is captured by the theory in~\Cref{fig:nondet-once}, rather than by $\thparam$.
  \todo[inline]{SM: Seems like a big gap that we don't know if the theories are the same. I have an idea so I'll come back to this.}
\end{remark}

\section{Summary and research directions}\label{sec:summ-rese-direct}

\todo[inline]{C: Discussion section about other applications of parameterized theories. Here or somewhere else.}

We have provided a fresh perspective on scoped effects in terms of the formalism of parameterized algebraic theories, using the idea that scopes are resources (Rem.~\ref{rem:scoped-signature-to-parameterized-signature}).
As parameterized algebraic theories have a sound and complete algebraic theory (Props.~\ref{thm:sound},~\ref{prop:free-models}), this carries over to a sound and complete equational theory for scoped effects.
We showed that our approach recovers the earlier models for scoped non-determinism, exceptions, and state (Thms.~\ref{thm:nondet-free-model}--\ref{thm:state-free-model}).

Here we have focused on equational theories for effects alone. But as is standard with algebraic effects, it is easy to add function types, inductive types, and so on, together with standard beta/eta theories (e.g.~\cite{pp-alg-op-gen-eff},\cite[\S5]{DBLP:conf/popl/Staton15}). This can be shown sound by the simple models considered here, as indeed the canonical model $\Set^\bN$ is closed and has limits and colimits.

\paragraph{Combinations of theories.}
A more concrete research direction is to define combinations of parameterized theories, and hence of scoped effects, using sum and tensor, generalizing combinations of algebraic effects~\cite{HPP06Combine}.
Informally, the tensor corresponds to all the operations from one theory commuting with the operations from the other.
But in the case of parameterized theories in $\Set^\bN$, defining what this means for operations to commute is delicate because substitution,~\cref{eq:general-subst-rule}, is more complex than in the algebraic case (but see~\cite[Def.~2]{DBLP:conf/lics/Staton13}).

For example, consider the commuting combination of explicit nondeterminism with $\once$~(\cref{fig:nondet-once}) and global state~(\Cref{sec:algebraic-effects}).
In the resulting theory, the state is rolled back when a computation fails.
For example, the fact that $\putop{0}$ commutes with $\fail$ means we can prove:
\begin{align*}
  \putop{1}\bigl(\orop(\putop{0}(\fail),\,\get(x_0,x_1)) \bigr) &= \putop{1}\bigl(\orop(\fail,\,\get(x_0,x_1)) \bigr) \\
                                                                &= \putop{1}\bigl(\get(x_0,x_1) \bigr) \\
  &= \putop{1}(x_1)
\end{align*}
using commutativity of $\putop{0}$ and $\fail$ for the first step.
The fact that $\putop{}$ and $\get$ commute with $\once$ and $\close$ means the state operations are independent of scopes.
\saml{Presumably the final term here should be $\putop{1}(x_1)$ rather than just $x_1$. SamS: I agree.}

\todo[inline]{C: an example involving once would be nicer.}

\medskip
Our work opens up new directions for scoped effects, in theory and in practice.
By varying the substructural laws of parameterized algebraic theories, we can recover foundations for scoped effects where scopes (as resources) can be reordered or discarded, i.e.~where they are not well-bracketed, already considered briefly in the literature~\cite{DBLP:conf/lics/PirogSWJ18}.
For example, the parameterized algebraic theory of qubits~\cite{DBLP:conf/popl/Staton15} might be regarded as a scoped effect, where we open a scope when a qubit is allocated and close the scope when it is discarded; this generalizes traditional scoped effects as multi-qubit operations affect multiple scopes.

\begin{acks}
We are grateful to many colleagues for helpful discussions and to the ESOP 2024 reviewers for their helpful comments and suggestions.
This work was supported by the UKRI Future Leaders Fellowship “Eﬀect Handler Oriented Programming” (reference number MR/T043830/1), ERC Project BLAST, and AFOSR Award No.~FA9550–21–1–003.

\todo[inline]{Reference ESOP paper, and say what is new, here or elsewhere.

  C: Stephanie said in an email:
``Your paper should be
accompanied by a cover letter stating that it is part of the ESOP
"journal-after" process. Your submission should also include a
detailed description of the added content to your paper and any other
changes that you have made to the work, especially in response
to your ESOP reviews.''

S: I would include this in the paper too because I've seen cases where the cover letter got overlooked. In our case the submission might look short so we can delicately say it's actually longer because it was a perspectives paper? }
\end{acks}

\bibliographystyle{ACM-Reference-Format}
\bibliography{refs}

\end{document}

%% file: preamble.tex
\usepackage{amsmath}  
\usepackage{stmaryrd}  
\usepackage{mathtools}
\usepackage{cleveref}
\usepackage{mathpartir}
\usepackage{tikz-cd}
\usepackage{tikz}
\usepackage{enumitem}
\usetikzlibrary{fit}

\usepackage[disable]{todonotes}
\usepackage{bbm}
\usepackage{wrapfig}

\newcommand{\vbar}{\mathrel{|}}

\newcommand{\op}{\mathrm{O}}

\newcommand{\yoneda}[1]{\mathbf{y}(#1)}

\newcommand{\bind}{\ensuremath{\mathbin{>\mkern-6.8mu>\mkern-6.9mu=}}}
\newcommand{\sdots}{\mathinner{{\ldotp}{\ldotp}{\ldotp}}} 
\newcommand{\den}[1]{\llbracket#1\rrbracket}
\newcommand{\Set}{\mathbf{Set}}

\newcommand{\catch}{\mathsf{catch}}
\newcommand{\throw}{\mathsf{throw}}
\newcommand{\close}{\mathsf{close}}
\newcommand{\orop}{\mathsf{or}}
\newcommand{\once}{\mathsf{once}}
\newcommand{\fail}{\mathsf{fail}}
\newcommand{\local}[1]{\mathsf{local}^{#1}}
\newcommand{\get}{\mathsf{get}}
\newcommand{\putop}[1]{\mathsf{put}^{#1}}

\newcommand{\openop}{\mathsf{openFile}}
\newcommand{\beginop}{\mathsf{begin}}
\newcommand{\opend}{\mathsf{end}}
\newcommand{\newop}[1]{\mathsf{new}^{#1}}

\newcommand{\cut}{\mathsf{cut}}
\newcommand{\scope}{\mathsf{scope}}
\newcommand{\concat}{\mathbin{+\!\!+}}%
\newcommand{\varcatch}{\mathsf{varcatch}}

\newcommand{\hide}[1]{}

\newcommand{\listf}{\mathsf{List}}
\newcommand{\idemf}{\mathsf{Idem}}

\newcommand{\tonce}{T_{\mathbf{o}}}
\newcommand{\tlocal}{T_{\mathbf{l}}}
\newcommand{\tcatch}{T_{\mathbf{c}}}
\newcommand{\tcut}{T_{\mathbf{cs}}}
\newcommand{\thparam}{\MCT_{\mathbf{param}}}
\newcommand{\freecut}{F_\mathbf{cs}}
\newcommand{\freeonce}{F_\mathbf{o}}
\newcommand{\freecatch}{F_\mathbf{c}}
\newcommand{\freelocal}{F_\mathbf{l}}

\newcommand{\MCC}{\ensuremath{\mathcal C}}

\newcommand{\MCR}{\ensuremath{\mathcal R}}

\newcommand{\MCT}{\ensuremath{\mathcal T}}

\newcommand{\MCX}{\ensuremath{\mathcal X}}
\newcommand{\MCY}{\ensuremath{\mathcal Y}}
\newcommand{\MCZ}{\ensuremath{\mathcal Z}}

\newcommand{\bN}{\ensuremath{\mathbb N}}

\newcommand{\ar}{\ensuremath{\mathit{ar}}}
\newcommand{\typing}[2]{#1\vdash#2}
\newcommand{\blank}{{-}}
\newcommand{\Tm}{\mathsf{Tm}}
\newcommand{\Bit}{\mathbbm{2}}
\newcommand{\defeq}{\coloneqq}
\newcommand{\later}{\mathop{\triangleright}}
\newcommand{\earlier}{\mathop{\triangleleft}}
\newcommand{\upp}[1]{{\upharpoonright\!#1}}
\newcommand{\downn}[1]{{\downharpoonright\!#1}}

\newcommand{\paren}[1]{\left(#1\right)}
\newcommand{\ver}[1]{\vert#1\vert}

\newcommand{\tuple}[1]{\langle#1\rangle}

\newcommand{\brc}[1]{\{#1\}}

\crefname{definition}{Def.\@}{Defs.\@}
\crefname{proposition}{Prop.\@}{Props.\@}

\newcommand{\eqth}[3]{\mathrel{=_{#1,(#2\vbar#3)}}}

\AtEndPreamble{%
\theoremstyle{acmdefinition}
\newtheorem{remark}[theorem]{Remark}
\theoremstyle{acmplain}  
\newtheorem{perspective}[theorem]{Perspective}
}
